\documentclass[11pt]{article}
\usepackage{amsfonts,amsthm,amssymb}
\usepackage[pdftex]{graphicx}
\usepackage{color}
\usepackage{url}
\usepackage[fleqn]{amsmath}
\usepackage{enumitem}
\usepackage[leftcaption]{sidecap}
\newcommand{\be}{\begin{eqnarray}}
\newcommand{\ee}{\end{eqnarray}}
\newcommand{\bez}{\begin{eqnarray*}}
\newcommand{\eez}{\end{eqnarray*}}
\renewcommand{\d}{\mathrm{d}}
\newcommand{\bd}{\bar{\mathrm{d}}}
\newcommand{\cA}{\mathcal{A}}
\newcommand{\bbN}{\mathbb{N}}

\theoremstyle{plain}
\newtheorem{theorem}{Theorem}[section]

\newtheorem{proposition}[theorem]{Proposition}

\theoremstyle{definition}

\newtheorem{remark}[theorem]{Remark}
\newtheorem{example}[theorem]{Example}

\numberwithin{equation}{section}
\numberwithin{theorem}{section}

\allowdisplaybreaks

\renewcommand{\theequation} {\arabic{section}.\arabic{equation}}

\begin{document}

\title{\bf Matrix Boussinesq solitons and their tropical limit}

\author{ \sc
 Aristophanes Dimakis$^1$, Folkert M\"uller-Hoissen$^{2,3}$, Xiao-Min Chen$^{2,4}$ \\ \small
 $^1$ Dept. of Financial and Management Engineering, University of the Aegean, Chios, Greece \\  \small
 $^2$ Max Planck Institute for Dynamics and Self-Organization, G\"ottingen, Germany \\ \small
 $^3$ Institute for Nonlinear Dynamics, Georg August University,
      37077 G\"ottingen, Germany \\ \small
 $^4$ College of Applied Sciences, Beijing University of Technology, Beijing 100124,
      People's Republic of China
}

\date{ }

\maketitle

\begin{abstract}
We study soliton solutions of matrix ``good'' Boussinesq equations, generated via a binary Darboux 
transformation. Essential features of these solutions are revealed 
via their ``tropical limit'', as also exploited in previous work about the KP equation. 
This limit associates a point particle interaction picture with a soliton (wave) solution. 
\end{abstract}

\section{Introduction}
\label{sec:intro}
The (scalar) Boussinesq equation originated from the study of water waves propagating in a canal 
\cite{Bous1872}. In this work we study the ``good'' Boussinesq equation, in which the highest 
derivative term has the opposite sign, as compared to the original Boussinesq equation. 
It also appears as a continuum limit of certain nonlinear atomic chains, see \cite{Pnev85}, for example. 
From a mathematical point of view, it belongs to the main examples of completely integrable PDEs. 
It is the second Gelfand-Dickey reduction \cite{Dick03} of the KP-II equation, the famous Korteweg-deVries (KdV) 
equation being the first. 
Compared with the latter, the behavior of its soliton solutions is considerably more diverse, 
in particular they can move in both directions, experience head-on collision, and solitons 
can split or merge (cf. \cite{Bogd+Zakh02,Rasi+Schi17}, also see the references cited in these papers). 
The scalar good Boussinesq equation is also known as ``nonlinear string equation'', see e.g. 
\cite{FST83,Lamb+Muse89}. 

In the KdV and Boussinesq case, a contour plot of a soliton solution displays a localization along a 
piecewise linear graph in two-dimensional space-time.
%(also see, e.g., \cite{Hiro+Ito83}). 
The definition of the ``tropical limit'' makes this precise and yields a method to compute it. 
Restricting the dependent variable to this graph, completes the tropical limit of the soliton solution, 
which displays the essentials of soliton interactions in a clear way. It is a very convenient tool 
to describe and to classify soliton solutions. 
The tropical limit associates with a soliton solution a point particle picture, also in the interaction 
region of solitons, revealing ``virtual solitons'' (borrowing a familiar notion from perturbative 
quantum field theory). 

In this work we address more generally the following $m \times n$ matrix potential Boussinesq equation, 
\be
  \phi_{tt} - 4 \beta  \phi_{xx} + \frac{1}{3} \phi_{xxxx} + 2 \left(\phi_x K \phi_x \right)_x
  - 2 \left( \phi_x K \phi _t  - \phi_t K \phi _x \right) = 0 \, ,   \label{matrixBSq_K}
\ee
where $K$ is a constant $n \times m$ matrix, $\beta >0$ (chosen as $\beta =1/4$ in  
all plots in this work), and a subscript indicates a partial derivative of $\phi$ with respect 
to the respective variable $x$ or $t$. We will refer to this equation as potential Bsq$_K$. 
It is the second Gelfand-Dickey reduction of the following KP-II equation, in a moving frame 
($x \mapsto x - 3 \beta t_3$, $t_2 =t$),
\bez
 -\frac{4}{3} \left(\phi_{t_3} + 3 \beta  \phi _x\right)_x + \phi_{tt} + \frac{1}{3} \phi_{xxxx}
 + 2 \left(\phi_x K \phi _x \right)_x
 - 2 \left( \phi_x K \phi_t - \phi_t K \phi_x \right) = 0 \, .
\eez
In terms of $u = 2 \phi_x$, we obtain the $m \times n$ matrix Boussinesq equation, or rather 
Bsq$_K$,\footnote{For the scalar Boussinesq equation, $x \mapsto -x$ and $t \mapsto -t$ are symmetries. 
In the matrix case, the first is still a symmetry, but $t \mapsto -t$ has to be accompanied by $u \mapsto u^T$.
}
\bez
  u_{tt} - 4 \beta  u_{xx} + \frac{1}{3} u_{xxxx} + (u K u)_{xx}
  - \left( u K \left(\partial^{-1}u\right)_t - \left(\partial^{-1}u\right)_t K u \right)_x = 0 \, .
\eez
 For a solution $u$ of the vector Boussinesq equation, where $n=1$ or $m=1$, $K u$, 
respectively $u K$, is a solution of the scalar Boussinesq equation. If $m,n>1$, we define the 
tropical limit graph via that of the scalar $\mathrm{tr}(K u)$, which is not in general a solution of the  
scalar Boussinesq equation. Our explorations and results fully substantiate this approach.  

Our analysis largely parallels that in \cite{DMH18}, where we explored line soliton solutions of the 
matrix KP-II equation in the tropical limit. There we concentrated on a class of solutions which we 
called ``pure solitons''. Another class has been treated in \cite{DMH18p}. Whereas pure solitons exhaust 
the solitons of the KdV reduction, this is not so for the Boussinesq reduction. 

Matrix versions of scalar integrable equations are natural extensions and of interest as models of 
coupled systems. A further motivation to explore them originated from the fact that in 2-soliton scattering,  
in particular in case of the matrix KdV \cite{Vese03,Gonc+Vese04} and the vector 
Nonlinear Schr\"odinger (NLS) equation \cite{Tsuch04,APT04IP}, in- and outgoing polarizations (values of 
the dependent variable attached to in- and outgoing solitons) are related 
by a Yang-Baxter map, a solution of the ``functional'' or ``set-theoretic'' version of the famous Yang-Baxter equation. 
For a matrix KP equation, a Yang-Baxter map is \emph{not} sufficient to fully describe the situation \cite{DMH18p}. 
It seems that we also have to go beyond Yang-Baxter in case of the second member  
in the family of Gel'fand-Dickey reductions of KP, which is the matrix Boussinesq equation. 

The (quantum) Yang-Baxter equation is well-known to express integrability in two-dimensional quantum field theory 
and exactly solvable models of statistical mechanics. In the present work, as also e.g. in 
\cite{DMH18,Vese03,Gonc+Vese04,Tsuch04,APT04IP}, we meet it in a classical context. Formally, however,
soliton waves may be regarded as a sort of quantization of the point particles constituting 
the tropical limit. 

In Section~\ref{sec:Bouss_bDT} we present a binary Darboux transformation for the potential Bsq$_K$
equation (\ref{matrixBSq_K}). Appendix~\ref{app:DBT} explains its origin from a general result in bidifferential 
calculus \cite{DMH13SIGMA,CDMH16}. In this work we concentrate on the case of vanishing seed solution. 
This leads to two cases, treated in Sections~\ref{sec:C'=-C} and \ref{sec:C'_noteq_-C}. 
The soliton solutions, obtained via the binary Darboux transformation, depend on parameters that have to 
be roots of a cubic equation. We introduce a convenient parametrization of these roots (see Section~\ref{subsec:param}) 
that greatly facilitates the further analysis. 
Section~\ref{sec:concl} contains some concluding remarks.

\section{A binary Darboux transformation for the matrix Boussinesq equation}
\label{sec:Bouss_bDT}
The following binary Darboux transformation is a special case of a general result in 
bidifferential calculus, see Appendix~\ref{app:DBT}. 
Let $N \in \bbN$. 
The integrability condition of the linear system
\be
   && \theta_t = \theta_{xx} + 2 \phi _{0,x} K \theta \, , \nonumber \\
   && \theta_{xxx} = 3 \beta \, \theta_x + \theta \, C - 3 \phi_{0,x} K \theta_x 
      - \frac{3}{2} \left( \phi_{0,t} + \phi_{0,xx} \right) K \theta \, ,   \label{Bsq_lin_sys}
\ee
where $\theta$ is an $m \times N$ and $C$ a constant $N \times N$ matrix, is the potential 
Bsq$_K$ equation for $\phi_0$. The same holds for the adjoint linear system
\be
   &&  \chi_t = -\chi_{xx} - 2 \chi K \phi_{0,x} \, , \nonumber \\
   &&  \chi_{xxx} = 3 \beta \, \chi_x + C' \chi - 3 \chi_x K \phi_{0,x}
      + \frac{3}{2} \chi K \left(\phi_{0,t} - \phi _{0,xx} \right) \, ,  \label{Bsq_adj_lin_sys}
\ee
where $\chi$ is an $N \times n$ and $C'$ a constant $N \times N$ matrix, in general different from $C$. 
The Darboux potential $\Omega$ satisfies the consistent $N \times N$ system of equations
\be
   && \Omega_x = -\chi K \theta \, , \nonumber \\
   && \Omega_t = - \chi K \theta_x + \chi_x K \theta \, , \nonumber \\
   && C' \Omega + \Omega C = - \chi K \theta_{xx} + \chi_x K \theta_x
      - \chi_{xx} K \theta  + 3 \beta \chi K \theta  
      - 3 \chi K \phi_{0,x} K \theta \, .     \label{Bsq_Omega}
\ee
At space-time points where $\Omega$ is invertible,  
\be
    \phi = \phi_0 - \theta \Omega^{-1} \chi    \label{BDT_new_solution}
\ee
is then a new solution of the potential Bsq$_K$ equation.

\begin{remark}
Taking the transpose of the above equations, and applying the substitution $t \mapsto -t$, 
we see that, besides $K \mapsto K^T$, we also obtain $C' \leftrightarrow C^T$ and 
$\theta \leftrightarrow \chi^T$. 
We also note that $x \mapsto -x$, $\phi_0 \mapsto - \phi_0$, 
$C \mapsto -C$, $C' \mapsto -C'$ is a symmetry of the linear systems. 
\end{remark}

\begin{remark}
\label{rem:BDT_transf}
The equations (\ref{Bsq_lin_sys}) - (\ref{BDT_new_solution}) are invariant under a transformation
\bez
    \theta \mapsto \theta \, A \, , \qquad 
    \chi \mapsto B \, \chi \, , \qquad
    C \mapsto A^{-1} \, C \, A \, , \qquad
    C' \mapsto B \, C' \, B^{-1} \, , \qquad
    \Omega \mapsto B \, \Omega \, A \, ,
\eez
with any invertible constant $N \times N$ matrices $A$ and $B$. 
\end{remark}

Using (\ref{BDT_new_solution}) and the first of (\ref{Bsq_Omega}), we find
\bez
    \mathrm{tr}( K \phi) &=& \mathrm{tr}(K \phi_0) - \mathrm{tr}(K \theta \, \Omega^{-1} \chi) 
                 = \mathrm{tr}(K \phi_0) - \mathrm{tr}(\chi K \theta \, \Omega^{-1})
                 = \mathrm{tr}(K \phi_0) + \mathrm{tr}(\Omega_x \, \Omega^{-1}) \nonumber \\
                 &=& \mathrm{tr}(K \phi_0) + (\log \det \Omega)_x \, .
\eez
Hence
\be
    \mathrm{tr}( K u ) - \mathrm{tr}(K u_0) = 2 \, (\log \det \Omega)_{xx} \, .
       \label{trKu-}
\ee
Such a formula is familiar in the scalar case, where $\det \Omega$ is the Hirota $\tau$-function. 
But we will see that, also in the matrix case, $\det \Omega$ plays a crucial role. In the following we will still 
call it $\tau$, after multiplication by a convenient factor, which preserves the relation (\ref{trKu-}).

\subsection{Solutions for vanishing seed}
The linear system with $\phi_0=0$ reads
\bez
   \theta_t = \theta_{xx} \, , \qquad
   \theta_{xxx} = 3 \beta  \theta_x + \theta C \, .
\eez
It possesses solutions of the form
\be
    \theta = \sum_a \theta_a \, e^{\vartheta(P_a)} \, ,   \label{theta}
\ee
where 
\be
     \vartheta(P) = P \, x + P^2 \, t \, ,   \label{Bouss_phase}
\ee
and each $P_a$ is a solution of the cubic equation
\be
    P_a^3 = 3 \beta P_a + C \, .   \label{P_sol_cubic_eq}
\ee
The index $a$ runs over any number of distinct roots.  

Correspondingly, the adjoint linear system takes the form
\bez
    \chi_t = -\chi_{xx} \, , \qquad
    \chi_{xxx} = 3 \beta \chi_x + C' \chi \, ,
\eez
which is solved by
\be
    \chi = \sum_{b} e^{- \vartheta(Q_b)} \, \chi_b    \label{chi}
\ee
if $Q_b$ solves the cubic equation
\be
    Q_b^3 = 3 \beta  Q_b - C' \, .  \label{Q_sol_cubic_eq}
\ee
The equations for the Darboux potential $\Omega$ are reduced to
\bez
 \Omega_x = -\chi K \theta \, , \quad
 \Omega_t = - \chi K \theta_x + \chi_x K \theta \, , \quad
 C' \Omega + \Omega C = - \chi K \theta_{xx} + \chi_x K \theta_x
      - \chi_{xx} K \theta  + 3 \beta \chi K \theta \, .
\eez
Writing
\be
   \Omega = \Omega _0 + \sum_{a,b} e^{-\vartheta(Q_b)} W_{ba} \, e^{\vartheta(P_a)} \, ,  \label{Omega_ansatz}
\ee
these equations are solved if $W_{ba}$ satisfies the Sylvester equation
\be
    Q_b W_{ba} - W_{ba} P_a = \chi_b K \theta_a \, ,   \label{Sylvester}
\ee
and if the constant matrix $\Omega_0$ is subject to
\be
    \Omega_0 C' + C \Omega_0 = 0 \, .   \label{Omega_0-Sylv_eq}
\ee
As a consequence of the last condition, there are two major cases. In Section~\ref{sec:C'=-C} 
we will address the case where $C' = -C$. Section~\ref{sec:C'_noteq_-C} then deals with 
the complement. 

We note that (\ref{trKu-}) reduces to
\be
    \mathrm{tr}(K u) = 2 \, (\log \det \Omega)_{xx} \, . \label{trKu}
\ee

\section{The case $C' = -C$}
\label{sec:C'=-C}
If $\Omega_0$ is invertible, Remark~\ref{rem:BDT_transf} shows that without restriction of generality we 
can choose $\Omega_0 = I_N$, the $N \times N$ identity matrix. (\ref{Omega_0-Sylv_eq}) then implies $C' = -C$. 
The remaining freedom of transformations, according to Remark~\ref{rem:BDT_transf}, is then given by  
transformations with $B = A^{-1}$. The similarity transformation $C \mapsto A^{-1} C A$ now allows us to 
assume that $C$ has Jordan normal form. 

$P_a$ and $Q_b$ are now solutions of the same cubic equation, so we can set 
\bez
      Q_a = P_a \, , 
\eez      
and the Sylvester equation takes the form 
\bez
    P_a W_{ab} - W_{ab} P_b = \chi_a K \theta_b \, .
\eez
If $a \neq b$ and $P_a$ and $P_b$ have disjoint spectrum, it is well-known that there is a solution and 
it is unique. In this case, the sum in the expression for $\theta$ or $\chi$ is over a disjoint set 
of solutions of the cubic equation. 

We will restrict our considerations to \emph{diagonal} matrices\footnote{A treatment of the case where $C$ contains 
larger than size 1 Jordan blocks is left aside in this work.}
\bez
    P_a = \mathrm{diag}(p_{1,a}, \ldots,p_{N,a}) \, , \qquad
    C = \mathrm{diag}(c_1,\ldots,c_N) \, . 
\eez
(\ref{P_sol_cubic_eq}) then requires
\be
    p_{ia}^3 = 3 \beta \, p_{ia} + c_i \qquad \quad i=1,\ldots,N \, .  \label{cubic_eq}
\ee
We will only consider \emph{real} roots. 
Writing
\be
  \chi_a = \left( \begin{array}{c}
            \eta_{1,a} \\ \vdots  \\ \eta_{N,a} 
                   \end{array}
            \right) \, , \qquad
  \theta_a = \left( \begin{array}{ccc}
            \xi_{1,a} & \cdots  & \xi_{N,a} 
                     \end{array}
             \right) \, ,            \label{chi,theta_decomp}
\ee
and $W_{ab} = (W_{ab,ij})$, we find
\bez
      W_{ab,ij} = \frac{\eta_{ia} K \xi_{jb}}{p_{ia}-p_{jb}} \qquad a \neq b \, , \quad
      p_{ia} \neq p_{jb} \, , \quad i,j=1,\ldots,N \, ,
\eez
and thus
\bez
  \Omega_{ij} = \delta_{ij} 
     + \sum_{ \substack{ a,b \\ a \neq b } } \frac{\eta_{ib} K \xi_{ja}}{ p_{ib}-p_{ja} } 
        \, e^{\vartheta(p_{ja})-\vartheta(p_{ib})} \qquad \quad i,j=1,\ldots,N  \, .
\eez

\subsection{A parametrization of the roots of the cubic equation}
\label{subsec:param}
We are only interested in \emph{real} soliton solutions, hence we restrict our analysis 
to \emph{real} roots of the cubic equation and demand that there are at least two different 
ones. This requires $|c_i| \leq 2 \beta^{3/2}$. 
We can then express the constants $c_i$ as follows,
\be
    c_i = 2 \beta ^{3/2} \frac{1 - 45 \lambda_i^2 + 135 \lambda_i^4 - 27 \lambda_i^6}{(1 + 3 \lambda_i^2)^3}
         \, ,  \label{c(lambda)}
\ee
where $\lambda_i$ are real parameters. 
The roots of the cubic equation (\ref{cubic_eq}) are then given by
\be
    p_{i,1}=-\frac{\sqrt{\beta } \left(1+6 \lambda_i-3 \lambda_i^2\right)}{1+3 \lambda_i^2} \, , \quad
    p_{i,2}=-\frac{\sqrt{\beta } \left(1-6 \lambda_i-3 \lambda_i^2\right)}{1+3 \lambda_i^2} \, , \quad
    p_{i,3}=\frac{2 \sqrt{\beta } \left(1-3 \lambda_i^2\right)}{1+3 \lambda_i^2}  \, .
       \label{roots}
\ee
All the roots satisfy $p^2 \leq 4 \beta$. We have $\lim_{|\lambda_i|\to\infty} p_{ia} = \sqrt{\beta}$ for $a=1,2$, 
and $\lim_{|\lambda_i|\to\infty} p_{i,3} = -2 \sqrt{\beta}$. Also see Fig.~\ref{fig:roots}.
The two involutive transformations 
\be
    \lambda \mapsto - \lambda \, , \qquad
    \lambda \mapsto \frac{1-\lambda}{1+3 \lambda}    \label{lambda_sym}
\ee
generate the permutation group of the three roots.

\begin{figure} 
\begin{center}
\includegraphics[scale=.4]{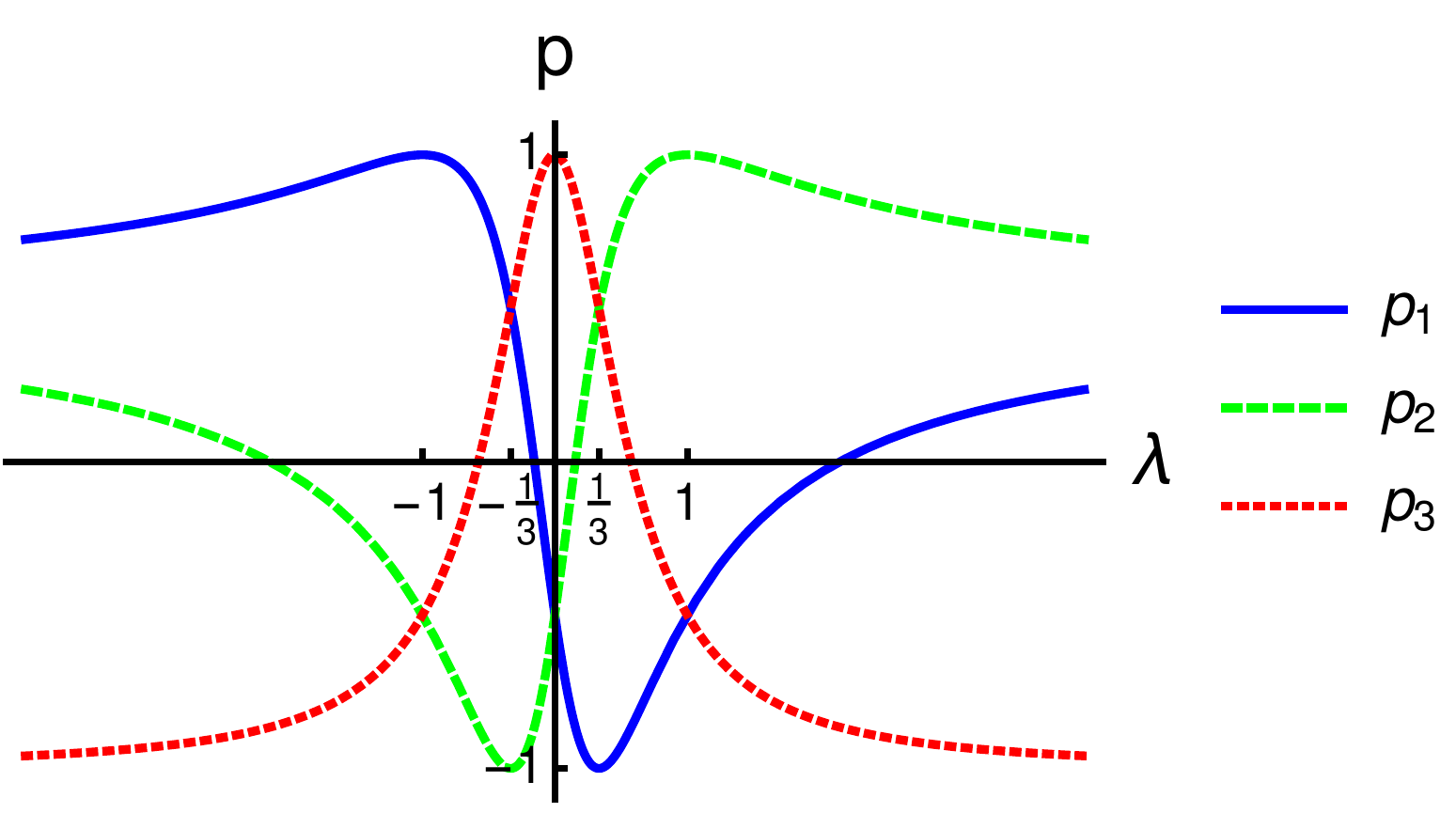} \hspace{.5cm}
\includegraphics[scale=.4]{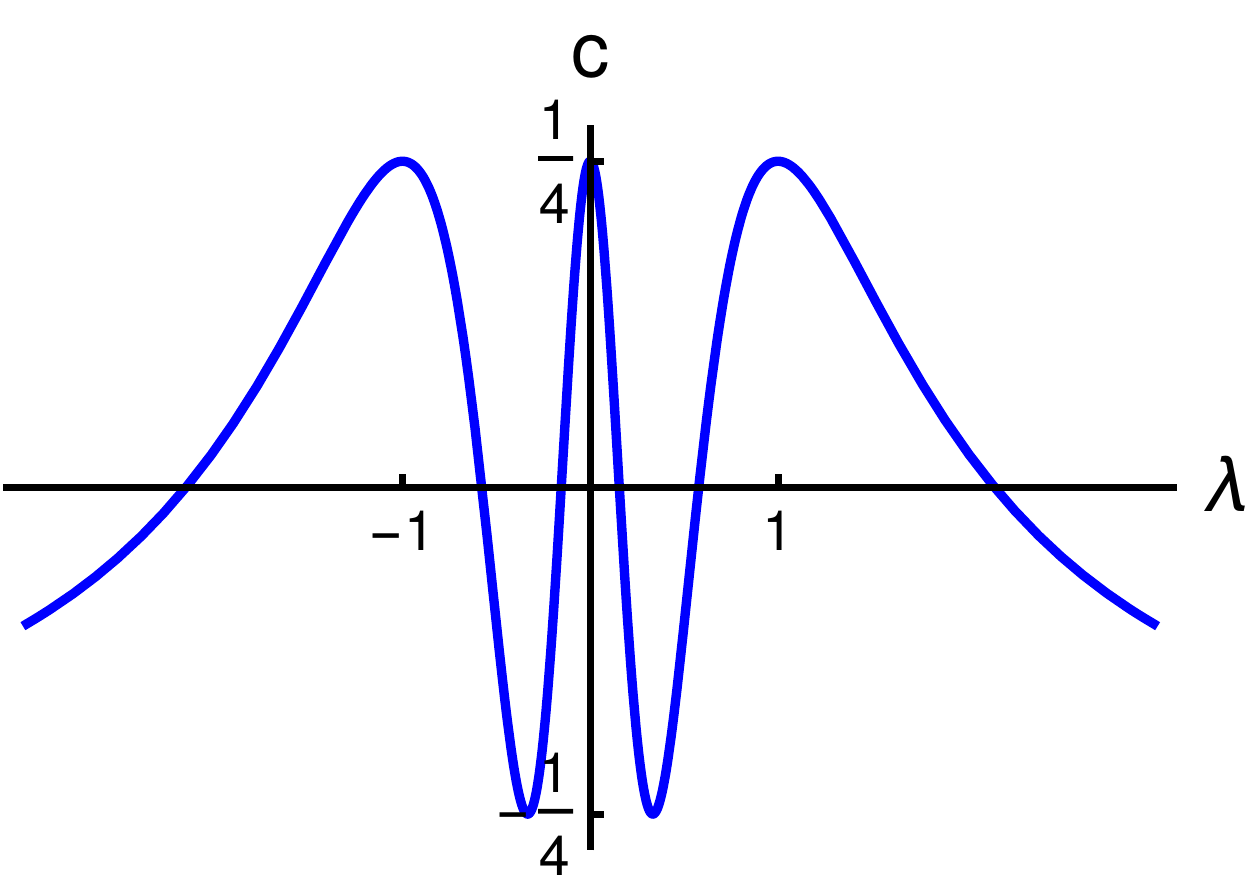}
\end{center}
\caption{The roots $p_a$ (the index $i$ in (\ref{roots}) is suppressed), $a=1,2,3$, of the cubic 
equation as functions of $\lambda$. The right plot shows $c$ as a function of $\lambda$, 
according to (\ref{c(lambda)}).
We chose $\beta = 1/4$, as also in the following figures. 
\label{fig:roots} }
\end{figure}

\subsection{Pure solitons} 
\label{subsec:pure_solitons}
In this subsection we select a subclass of soliton solutions, which we call ``pure solitons''. 
The main characterization is to restrict the expressions for the solutions of the linear systems 
in (\ref{theta}) and (\ref{chi}) to only involve a single root of each of the cubic equations 
(\ref{P_sol_cubic_eq}) and (\ref{Q_sol_cubic_eq}). 
The latter equations coincide, since we assume $C'=-C$ in this section, so we have to choose two different roots, 
$P_1 = P$ and $P_2 = Q$, of the same cubic equation. We will further assume that these matrices are diagonal, hence
\bez
  && P = \mathrm{diag}(p_1,\ldots,p_N) =: \mathrm{diag}(p_{1,1},\ldots,p_{N,1})  \, , \\
  && Q = \mathrm{diag}(q_1,\ldots,q_N) =: \mathrm{diag}(p_{1,2},\ldots,p_{N,2}) \, .
\eez
It will be convenient to allow both notations for the (diagonal) entries. 
The constants $p_i,q_i$ have to solve the cubic equation $z^3 = 3 \beta z + c_i$. We will require $q_i \neq p_i$, $i=1,\ldots,N$. Moreover, we will mostly 
also assume $p_i \neq p_j$ and $q_i \neq q_j$ for $i \neq j$. Writing
\bez
  \chi_1 = \left( \begin{array}{c}
            \eta_1 \\ \vdots  \\ \eta_N 
                   \end{array}
            \right) \, , \qquad
  \theta_1 = \left( \begin{array}{ccc}
            \xi_1 & \cdots  & \xi_N 
                     \end{array}
             \right) (Q-P) \, ,    
\eez
we then have
\bez
    \Omega_{ij} = \delta_{ij} + w_{ij} \, e^{\vartheta(p_j)-\vartheta(q_i)} \, , \qquad
     w_{ij} = \frac{q_j-p_j}{q_i-p_j} \, \eta_i K \xi_j  \, .
\eez
This is exactly the expression for $\Omega$ that we found in \cite{DMH18} for the KP$_K$ equation.
The only difference is that, for $i=1,\ldots,N$, $q_i \neq p_i$ now have to satisfy the cubic equations, so that 
$p_i^3 - 3 \beta  p_i = q_i^3 - 3 \beta  q_i$ holds, which is
\bez
      p_i^2 + p_i q_i + q_i^2 = 3 \beta \, .
\eez
But, apart from this, formulae derived in \cite{DMH18} for the potential KP$_K$ equation also apply to 
the case under consideration. Next we summarize those that are needed in this work. 
Introducing
\be
    \tau = e^{\vartheta(q_1) + \cdots + \vartheta(q_N)} \, \det \Omega \, , \qquad
    F = - e^{\vartheta(q_1) + \cdots + \vartheta(q_N)} \, \theta \, e^{\vartheta(P)} 
    \mathrm{adj}(\Omega) e^{-\vartheta(Q)} \, \chi \, ,   \label{pure_tau_F}
\ee
where $\mathrm{adj}(\Omega)$ is the adjugate of the matrix $\Omega$, 
we find that they have expansions
\bez
    \tau = \sum_{I \in \{1,2\}^N} \tau_I \, , \qquad
    \tau_I := \mu_I \, e^{\vartheta_I} \, , \qquad
    F = \sum_{I \in \{1,2\}^N} M_I \, e^{\vartheta_I} \, ,
\eez
where $\mu_I$ are constants, $M_I$ constant matrices, and 
\bez
    \vartheta_I = \sum_{k=1}^N \vartheta(p_{k \,a_k}) \qquad I=(a_1,\ldots,a_N) \, , \quad a_k \in \{1,2\} \, . 
\eez
Recall that $p_{k,1} = p_k$ and $p_{k,2} = q_k$, $k=1,\ldots,N$. We have
\bez
     \phi = \frac{F}{\tau }  \, .
\eez
Since $\tau_{I,x} = p_I \tau_I$, where $p_I = p_{1,a_1} + \cdots + p_{N,a_N}$ if $I=(a_1,\ldots,a_N)$, 
we obtain
\bez
     u = \frac{1}{\tau^2} \sum_{I,J \in \{1,2\}^N} (p_I  - p_J) (\mu_J M_I - \mu_I M_J) \, e^{\vartheta_I + \vartheta_j} \, .
\eez
Note that (cf. \cite{DMH18})
\bez
     \mathrm{tr} (K M_I) = \Big( p_I - \sum_{i=1}^N q_i \Big) \, \mu_I  \, ,
\eez
so that
\bez
     \mathrm{tr} (K u) = 2 (\log \tau)_{xx} \, ,
\eez
in accordance with (\ref{trKu}). 

If $\tau \neq 0$ and if $\mu_I>0$ for all $I$ with $\mu_I \neq 0$ in the expression for $\tau$, regularity of 
the solution is guaranteed. Let
\bez
    \mathcal{U}_I = \left\{ (x,t)\in \mathbb{R}^2 \, | \, \log \tau_I \geq \log \tau_J, \, J \in \{1,2\}^N \right\}
    \, .
\eez
We call this the region where $\vartheta_I$ dominates. As intersection of half-spaces, it is convex. 

If $\mu_I >0$, the tropical limit of $\phi$ in $\mathcal{U}_I$ is\footnote{Appendix~D in \cite{DMH11KPT} 
explains the relation with the tropical limit defined via the Maslov dequantization formula.} 
\bez
        \phi_I := \frac{M_I}{\mu_I} \, .
\eez  
For $I \neq J$, the intersection $\mathcal{U}_I \cap \mathcal{U}_J$ is a segment 
of the straight line determined by $\log \tau_I = \log \tau_J$. On such a (visible) segment the 
value of $u$ is given by
\bez
    u_{IJ} = \frac{1}{2} (p_I -p_J) (\phi_I - \phi_J) \, .
\eez
We find $\mathrm{tr}(K u_{IJ}) = \frac{1}{2} (p_I-p_J)^2$ and introduce normalized values
\bez
    \hat{u}_{IJ} = \frac{\phi_I-\phi_J}{p_I-p_J} \, ,
\eez
which satisfy $\mathrm{tr}(K \hat{u}_{IJ}) = 1$.
Using the notation
\bez
    I_k(a) = \left( a_1,\ldots,a_{k-1},a,a_{k+1},\ldots,a_N \right) \, ,
\eez
the $k$-th soliton appears in space-time on segments of the straight lines determined by 
$\log \tau_{I_k(1)} = \log \tau_{I_k(2)}$, i.e.,
\be
   x + (p_k+q_k) \, t + \frac{1}{p_k-q_k}  \log \frac{\mu_{I_k(1)}}{\mu_{I_k(2)}} = 0 \, .
       \label{k-soliton_line}
\ee
This also determines the asymptotic structure of a tropical limit graph of a pure $N$-soliton 
solution. Without restriction of generality, we can order the parameters such that 
$p_1 + q_1 < p_2 + q_2 < \cdots < p_N + q_N$.\footnote{With the chosen parametrization, 
$p_i+q_i = p_j+q_j$, for some $i,j$, implies either $p_i=p_j$ and $q_i=q_j$, or $p_i = q_j$ and $p_j = q_i$. 
We exclude these cases. }
If we represent 
$p_i$ and $q_i$ by the first two roots in (\ref{roots}), then $p_i+q_i$ is a strictly 
increasing function of $|\lambda|$, hence this order is obtained by choosing 
$0 < \lambda_1 < \lambda_2 < \cdots < \lambda_N$. 
Now it follows from (\ref{k-soliton_line}) that, for $t \ll 0$, the solitons appear along the 
$x$-axis according to their numbering, and for $t \gg 0$ they appear in reverse order. 
To the left of the center is the dominating phase region $\mathcal{U}_{1,\ldots,1}$.
Then follows counterclockwise $\mathcal{U}_{2,1,\ldots,1}$, $\mathcal{U}_{2,2,1,\ldots,1}$, ..., 
unless we get to the region $\mathcal{U}_{2,\ldots,2}$ on the right hand side. Correspondingly, 
starting again from $\mathcal{U}_{1,\ldots,1}$, we get clockwise to the regions 
$\mathcal{U}_{1,\ldots,2}$, $\mathcal{U}_{1,\ldots,2,2}$, ..., until we arrive at 
$\mathcal{U}_{2,\ldots,2}$. These regions always appear in a tropical limit graph of a pure 
$N$-soliton solution. The remaining $\mathcal{U}_I$ can only appear as \emph{bounded} 
regions. But some may be empty. This depends on the value of higher Boussinesq hierarchy variables, 
also see Fig.~\ref{fig:3-soliton} below.

\begin{remark}
It can happen that $\mu_I =0$ for some multi-index $I$, so that $e^{\vartheta_I}$ is absent 
in the expression for $\tau$, but that this exponential appears in the numerator of the 
expression for $u$, i.e., $M_I \neq 0$. Then some components of $u$ will exhibit exponential growth 
in the region where this phase dominates, and a tropical limit does not exist. 
\end{remark}

\subsubsection{Single soliton solution}
For the 1-soliton solution we find
\bez
    \tau = e^{p x + p^2 t+\varphi_0} + e^{q x+q^2 t-\varphi_0} \, , \qquad
    \varphi_0 = \frac{1}{2} \log (\eta K \xi) \, ,
\eez
assuming $\eta K \xi >0$, and 
\bez
    \phi = (p-q) 
     \frac{e^{p x+p^2 t+\varphi_0}}{e^{p x+p^2 t+\varphi_0}+e^{q x+q^2 t-\varphi _0}}
     \frac{\xi \otimes \eta}{\eta  K \xi} \, .
\eez
This yields
\bez
   u = \frac{1}{2} (p-q)^2 
    \mathrm{sech}^2 \left(\frac{1}{2} (p-q) (x+(p+q) t)+\varphi_0\right) 
    \frac{\xi \otimes \eta}{\eta  K \xi}  \, .
\eez
Setting 
\bez
    p = -\sqrt{\beta} \frac{ 1 + 6 \lambda - 3 \lambda^2}{1 + 3 \lambda^2} \, , \qquad
    q = - \sqrt{\beta} \frac{ 1 - 6 \lambda -3 \lambda^2}{1 + 3 \lambda^2} \, ,
\eez
it takes the form
\bez
  u = 4 \beta  \frac{18 \lambda^2}{(1 + 3 \lambda^2)^2} 
    \mathrm{sech}^2 \left( 2 \sqrt{\beta} \frac{ 3 \lambda}{1+3 \lambda^2} 
     \left(x - 2 \sqrt{\beta} \frac{ 1-3 \lambda^2}{1+3 \lambda^2} t \right) + \varphi_0 \right) 
      \frac{\xi \otimes \eta}{\eta  K \xi} \, .
\eez
Via the symmetries (\ref{lambda_sym}), 1-soliton solutions with other choices of the roots 
are obtained from the above solution. 
If $\lambda^2 < 1/3$, the soliton moves from left to right. If $\lambda^2 > 1/3$, it
moves from right to left. For $\lambda^2 = 1/3$, it is stationary.
In all cases, the absolute value of the velocity is less than $2 \sqrt{\beta}$.
We also note that $0 \leq \mathrm{tr}(K u) \leq 6 \beta$. 

The tropical limit graph of the 1-soliton solution is the boundary between the two dominating 
phase regions $\mathcal{U}_1$ and $\mathcal{U}_2$. It is the straight line in space-time 
($xt$-plane), determined by 
\bez 
   x + (p+q) t + \frac{1}{p-q} \log (\eta K \xi) = 0 \, ,
\eez
with slope $-1/(p+q)$.
We have
\bez
    u_{1,2} = \frac{1}{2} (p-q)^2 \, \frac{\xi \otimes \eta}{\eta  K \xi} \, , \qquad
    \hat{u}_{1,2} = \frac{\xi \otimes \eta }{\eta  K \xi} \, .
\eez
   
\subsubsection{2-soliton solution}
\label{subsec:pure_2-soliton}
In this case ($N=2$), we find
\bez
    \tau = \alpha \, e^{\vartheta(p_1) + \vartheta(p_2)} 
     + \kappa_{1,1} \, e^{\vartheta(p_1) + \vartheta(q_2)}
     + \kappa_{2,2} \, e^{\vartheta(p_2) + \vartheta(q_1)}
     + e^{\vartheta(q_1) + \vartheta(q_2)} \, ,
\eez
where
\be
 &&    \kappa_{ij} = \eta_i \, K \, \xi_j \, ,   \label{kappa_ij} \\
 && \alpha =  \kappa_{1,1} \kappa_{2,2} - \frac{(q_1-p_1) (q_2-p_2)}{(q_1-p_2) (q_2-p_1)}  \, \kappa_{1,2} \kappa_{2,1} \, ,
    \nonumber
\ee
and
\bez
   F &=& (q_1-p_1) (q_2-p_2) \Big( \frac{\kappa_{2,2}}{p_2-q_2} \xi_1 \otimes \eta_1
       + \frac{\kappa_{1,1}}{p_1-q_1} \xi_2 \otimes \eta_2 + \frac{\kappa_{1,2}}{q_1-p_2} \xi_1 \otimes \eta_2 \\
    && + \frac{\kappa_{2,1}}{q_2-p_1} \xi_2\otimes \eta_1 \Big) e^{\vartheta(p_1) + \vartheta(p_2)} 
       + (p_1-q_1) \xi_1 \otimes \eta_1 \, e^{\vartheta(p_1) + \vartheta(q_2)}
       + (p_2-q_2) \xi_2 \otimes \eta_2 \, e^{\vartheta(p_2) + \vartheta(q_1)} \, .
\eez
Hence, if $\alpha, \kappa_{1,1}, \kappa_{2,2} \neq 0$,
\bez
 && \hspace*{-.5cm} \phi_{1,1} = \frac{(q_1-p_1) (q_2-p_2)}{\alpha} \Big(
         \frac{ \kappa_{2,2}}{p_2-q_2} \, \xi_1 \otimes \eta_1 
       + \frac{\kappa_{1,2}}{q_1-p_2} \, \xi_1 \otimes \eta_2 
       + \frac{\kappa_{2,1}}{q_2-p_1} \, \xi_2 \otimes \eta_1 
       + \frac{\kappa_{1,1}}{p_1-q_1} \, \xi_2 \otimes \eta_2 \Big) \, , \\
 && \hspace*{-.5cm} \phi_{1,2} = \frac{p_1-q_1}{\kappa_{1,1}} \, \xi_1 \otimes \eta_1 \, , \quad
  \phi_{2,1} = \frac{p_2-q_2}{\kappa_{2,2}} \, \xi_2 \otimes \eta_2 \, , \quad 
  \phi_{2,2} = 0 \, .
\eez
Choosing
\be
    p_i = - \sqrt{\beta} \frac{1 + 6 \lambda _i - 3 \lambda_i^2}{1 + 3 \lambda_i^2} \, , \qquad
    q_i = - \sqrt{\beta} \frac{1 - 6 \lambda _i - 3 \lambda_i^2}{1 + 3 \lambda_i^2} \, ,
          \label{Bouss_p,q__2-soliton}
\ee
we obtain via $\phi = F/\tau$ a 2-soliton solution of the potential Bsq$_K$ 
equation.\footnote{If $\alpha$, $\kappa_{1,1}$ 
or $\kappa_{2,2}$ vanishes, the tropical limit may still be defined. But since the corresponding 
phase is absent in $\tau$, there is then no value $\phi_{1,1}$, $\phi_{1,2}$, respectively $\phi_{2,1}$.}
Fig.~\ref{fig:2-soliton} shows examples of corresponding tropical limit graphs. Applying the 
symmetries (\ref{lambda_sym}), 2-soliton solutions with other choices of the roots are 
obtained from the above solution.

\begin{figure} 
\begin{center}
\begin{minipage}{0.04\linewidth}
\vspace*{.15cm}
\bez
 \begin{array}{cc} t & \\ \uparrow & \\ & \rightarrow x \end{array}                
\eez
\end{minipage}
\includegraphics[scale=.2]{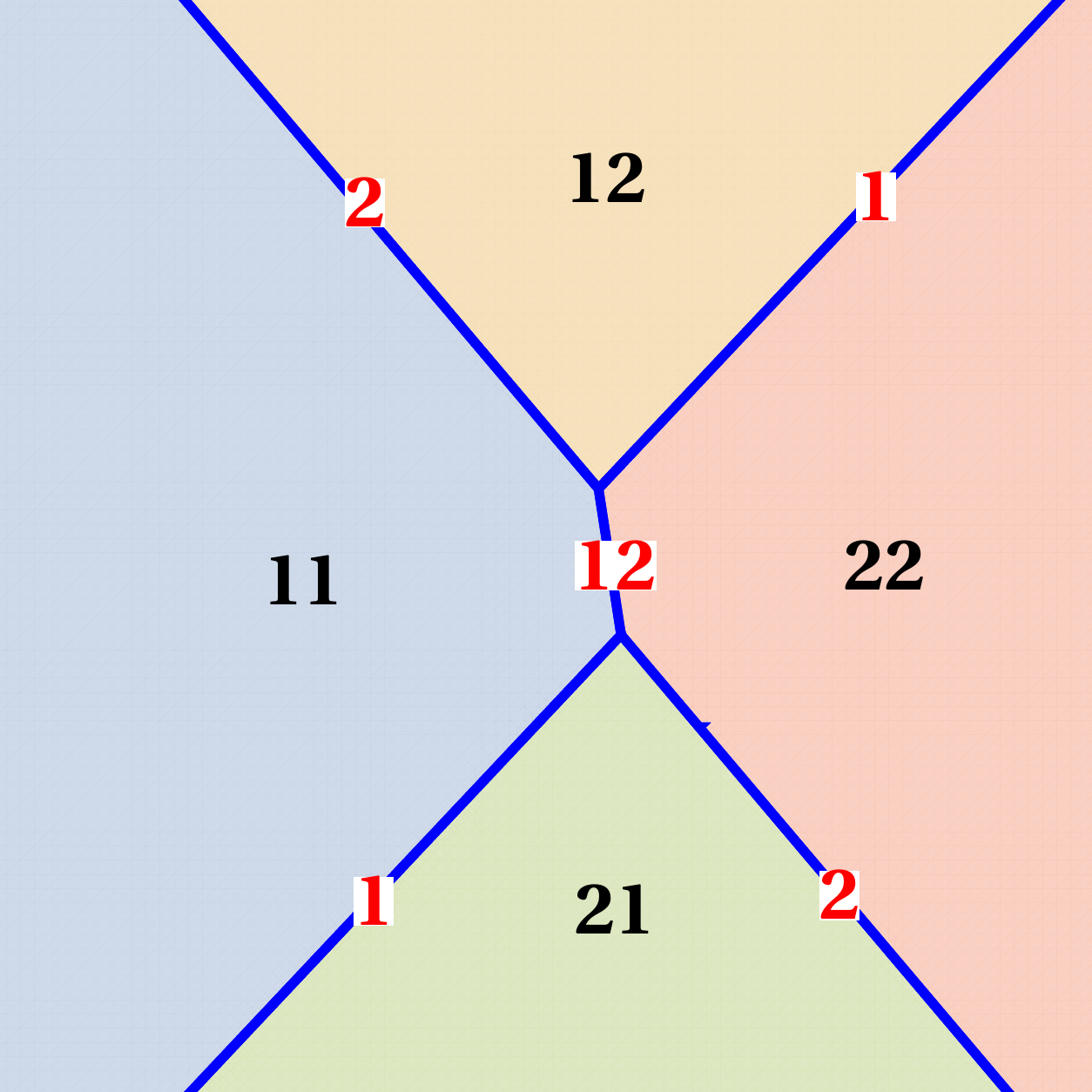} 
\hspace{.3cm}
\includegraphics[scale=.2]{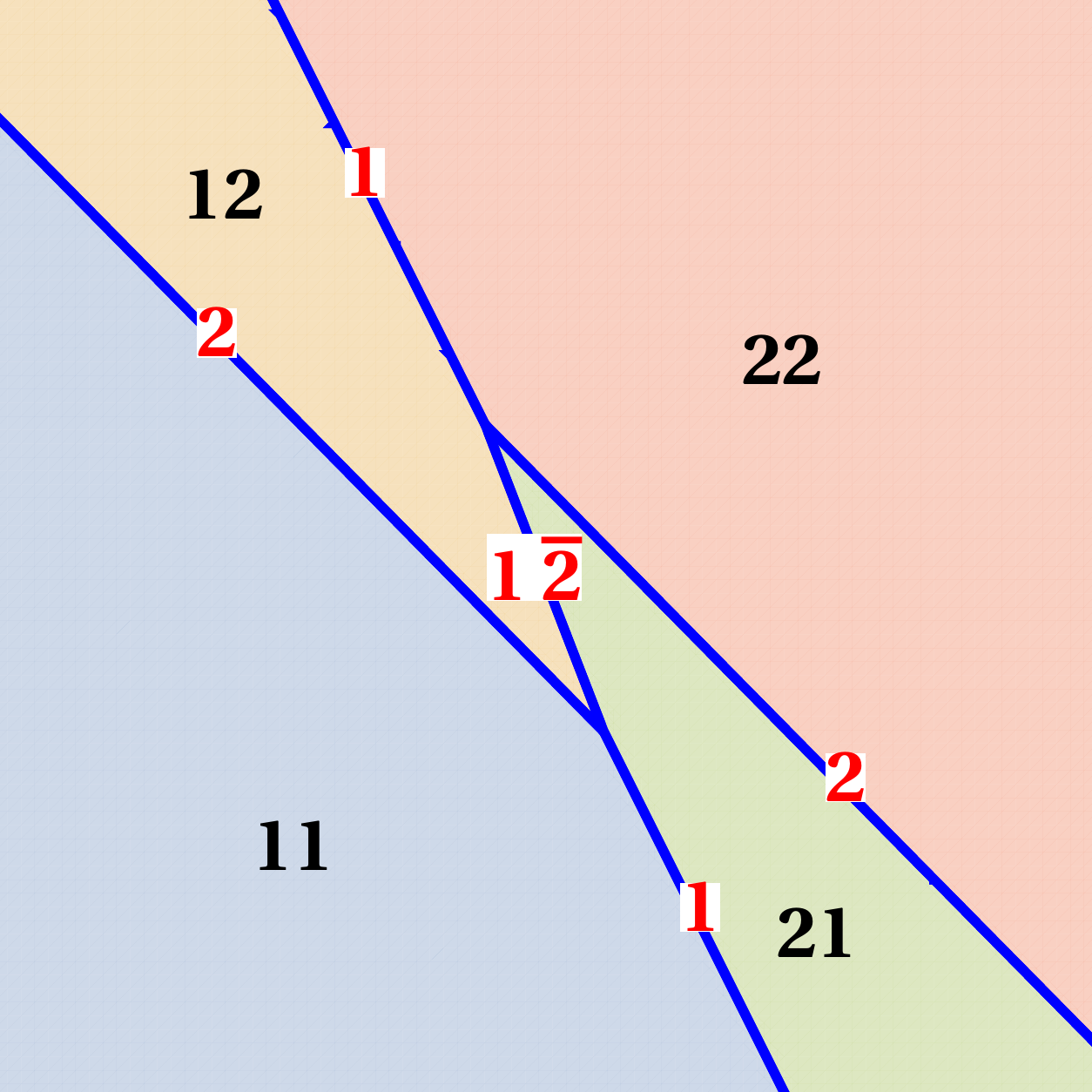} 
\hspace{.3cm}
\includegraphics[scale=.2]{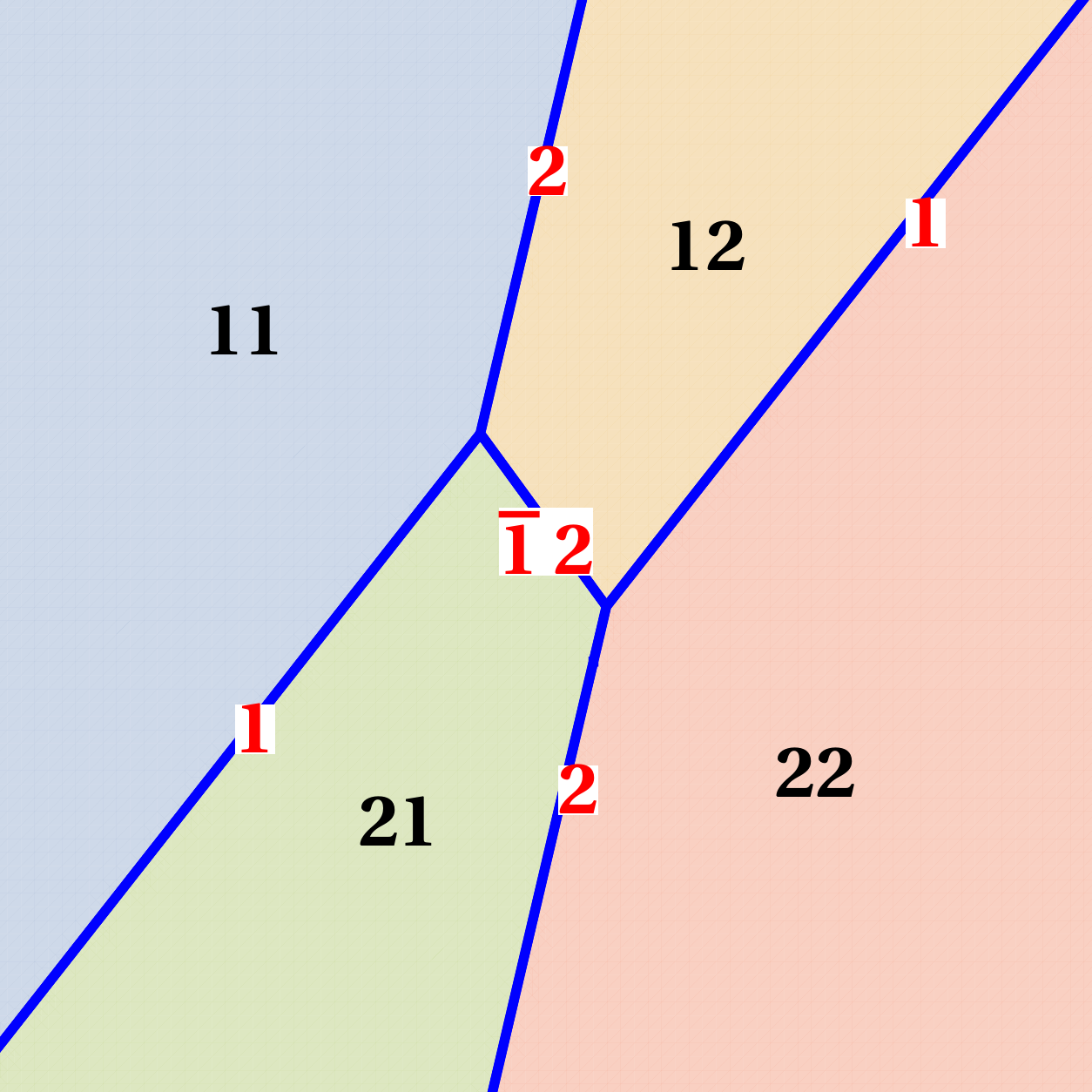} 
\end{center}
\caption{Tropical limit graphs of 2-soliton solutions of the $m=2$ vector (i.e., $n=1$) Bsq$_K$ equation 
in space-time ($xt$-plane). Here we chose $K=(1,1)$, $\eta_1 = \eta_2 =1$, 
$\xi_1 = (1,0)^T$, $\xi_2 = (0,1)^T$, and the parameter values 
$(\lambda_1,\lambda_2) = (1/10,2), (7/10,4),(1/5,5/11)$.  
This covers all possible directions of motion and interactions of two solitons. 
In the first case, the two solitons, numbered by 1 and 2, are moving towards each 
other (head-on collision), then form an intermediate ``virtual'' soliton 
(cf. \cite{Hiro+Ito83,Lamb+Muse89}), here denoted 
as a ``composite'' 12, which finally splits. This has no counterpart in the KdV case. 
In the other cases, the numbering of the virtual soliton draws on 
a formal analogy with particles and anti-particles (the latter indicated by a bar over 
the respective number). 
\label{fig:2-soliton} }
\end{figure}

Furthermore, we find
\bez
  \hat{u}_{11,12} &=& \frac{1}{\alpha} \Big( 
    \frac{\kappa_{1,2} \, \kappa_{2,1}}{\kappa_{1,1}} \frac{(p_1-q_1)^2}{(p_1-q_2)(p_2-q_1)} \, \xi_1 \otimes \eta_1
  - \kappa_{1,2} \, \frac{p_1-q_1}{p_2-q_1} \, \xi_1 \otimes \eta_2  \\
 && - \kappa_{2,1} \, \frac{(p_1-q_1)(p_2-q_1)}{(p_1-q_2)(p_2-q_1)} \, \xi_2 \otimes \eta_1 
  + \kappa_{1,1} \, \xi_2 \otimes \eta_2 \Big) \, , \\
  \hat{u}_{11,21} &=& \frac{1}{\alpha} \Big( \kappa_{2,2} \, \xi_1 \otimes \eta_1 
   - \kappa_{1,2} \, \frac{p_2-q_2}{p_2-q_1} \, \xi_1 \otimes \eta_2 
   - \kappa_{2,1} \, \frac{(p_2-q_1)(p_2-q_2)}{(p_1-q_2)(p_2-q_1)} \, \xi_2 \otimes \eta_1  \\
 && + \frac{\kappa_{1,2} \, \kappa_{2,1}}{\kappa_{2,2}} 
         \frac{(p_2-q_2)^2}{(p_1-q_2)(p_2-q_1)} \, \xi_2 \otimes \eta_2 \Big) \, , \\
  \hat{u}_{12,22} &=& \frac{1}{\kappa_{1,1}} \, \xi_1 \otimes \eta_1 \, , \qquad
   \hat{u}_{21,22} = \frac{1}{\kappa_{2,2}} \, \xi_2 \otimes \eta_2  \, .
\eez
The map $\mathcal{R} : \, (\hat{u}_{11,21},\hat{u}_{21,22}) \mapsto (\hat{u}_{12,22},\hat{u}_{11,12})$ is a 
nonlinear Yang-Baxter map, it satisfies the (quantum) Yang-Baxter equation 
\be
  \mathcal{R}_{12} \circ \mathcal{R}_{13} \circ \mathcal{R}_{23}
  = \mathcal{R}_{23} \circ \mathcal{R}_{13} \circ \mathcal{R}_{12} \, .  \label{YB}
\ee  
This map is also obtained from that of the KP theory, see Section~5 in \cite{DMH18}, with the 
parameters restricted by (\ref{Bouss_p,q__2-soliton}). 
The fact that $\mathcal{R}$ satisfies the Yang-Baxter equation follows from the Yang-Baxter property 
of the more general KP$_K$ Yang-Baxter map. Alternatively, it follows from the 3-soliton solution 
and the fact that the polarizations do not depend on the independent variables, including higher 
hierarchy variables, see Section~\ref{subsec:3solitons}.

In the vector case, we obtain a \emph{linear} map,  
\bez
     (\hat{u}_{12,22},\hat{u}_{11,12}) = (\hat{u}_{11,21},\hat{u}_{21,22}) \, R(\lambda_1, \lambda_2) \, ,
\eez
where
\bez
    R(\lambda_i,\lambda_j) = \left(\begin{array}{cc}
 \frac{\lambda_i-\lambda_j}{\lambda_i+\lambda_j} 
 \frac{1-\lambda_i-\lambda_j -3 \lambda_i \lambda_j}{1 -\lambda_i+\lambda_j + 3 \lambda_i \lambda_j} & 
 \frac{2 \lambda_i}{\lambda_i+\lambda_j} \frac{1 + 3 \lambda_j^2}{1 -\lambda_i+\lambda_j + 3 \lambda_i \lambda_j} \\[.5em]
 \frac{2 \lambda_j}{\lambda_i+\lambda_j} \frac{1 + 3 \lambda_i^2}{1 -\lambda_i+\lambda_j + 3 \lambda_i \lambda_j} & 
 \frac{\lambda_j-\lambda_i}{\lambda_i+\lambda_j} \frac{1+\lambda_i+\lambda_j-3 \lambda_i \lambda_j}{1-\lambda_i+\lambda_j + 3 \lambda_i \lambda_j}
\end{array}\right) \, .
\eez

\begin{remark}
Dropping the exponential factor in (\ref{pure_tau_F}), which has only been introduced to achieve a 
convenient numbering of phases, in the $N=2$ case we obtain 
\bez
   \tilde{\tau} 
 = 1 + e^{\zeta_1} + e^{\zeta_2} + \left( \frac{\kappa_{1,1} \kappa_{2,2} - \kappa_{1,2} \kappa_{2,1}}
        {\kappa_{1,1} \kappa_{2,2}}  
   + \frac{\kappa_{1,2} \kappa_{2,1}}{\kappa_{1,1} \kappa_{2,2}} \, 
       \frac{(p_1-p_2) (q_1-q_2)}{(p_1-q_2) (q_1-p_2)} \right) \, e^{\zeta_1+\zeta_2} \, ,
\eez
where $\zeta_i = \vartheta(p_i) - \vartheta(q_i) + \log \kappa_{ii}$, assuming $\kappa_{ii} >0$. 
Comparison with a known expression for the $\tau$-function of the 2-soliton solution of the scalar 
Boussinesq (or KP) equation shows that this determines a solution of the scalar 
Boussinesq equation if $\kappa_{1,1} \kappa_{2,2} = \kappa_{1,2} \kappa_{2,1}$. We also note that, 
if $\kappa_{1,2} \kappa_{2,1} = 0$, the above expression factorizes to 
$\tilde{\tau} = (1 + e^{\zeta_1}) (1 + e^{\zeta_2})$. In this case, the tropical limit graph 
is simply the superimposition of those of the factors\footnote{This is evident from 
the Maslov dequantization formula, cf. Appendix~D in \cite{DMH11KPT}, for example.}, hence there 
is no phase shift. 
\end{remark}

\subsubsection{Degenerations of the pure 2-soliton solution of the vector Boussinesq equation}
We consider the special cases where $p_1 = p_2$ or $q_1 = q_2$. Then we have $c_1 = c_2$, so that all 
parameters are roots of a single cubic equation. 
Since we represent $p_i$ and $q_i$ by the 
first two roots in (\ref{roots}), this means that $\lambda_2 = (1-\lambda_1)/(3 \lambda_1+1)$, 
respectively $\lambda_2 = (1+\lambda_1)/(3 \lambda_1-1)$. In both cases we have
$\alpha = \kappa_{1,1} \kappa_{2,2} - \kappa_{1,2} \kappa_{2,1}$, which vanishes when we address the vector 
Boussinesq equation. If $q_2 = q_1$, we obtain
\bez
  \tau = \Big( e^{\vartheta(q_1)} + \kappa_{1,1} \, e^{\vartheta(p_1)} 
          + \kappa_{2,2} \, e^{\vartheta(p_2)} \Big) \, e^{\vartheta(q_1)} \, .
\eez
The factor $e^{\vartheta(q_1)}$ does not influence the tropical limit graph, which is shown in 
Fig.~\ref{fig:pure_2-soliton_degen}. 
If $p_2 = p_1$, we find
\bez
 \tau = \Big( e^{-\vartheta(p_1)} + \kappa_{1,1} \,  e^{-\vartheta(q_1)}
 + \kappa_{2,2} \, e^{-\vartheta(q_2)} \Big) \, e^{\vartheta(p_1) + \vartheta(q_1) + \vartheta(q_2)} \, . 
\eez
The corresponding tropical limit graphs are Y-shaped (a soliton splits into two), respectively reverse Y-shaped
(two solitons merge), see Fig.~\ref{fig:pure_2-soliton_degen}. Approaching such a solution by 
letting $q_2 \to q_1$, respectively $p_2 \to p_1$, in the 2-soliton solution in Section~\ref{subsec:pure_2-soliton},
we see that the edge representing the virtual soliton (cf. the second and third graph in Fig.~\ref{fig:2-soliton})
gets longer and longer, in such a way that the dominating phase region $\mathcal{U}_{1,1}$ finally disappears at infinity. 
\begin{figure} 
\begin{center}
\begin{minipage}{0.04\linewidth}
\vspace*{.15cm}
\bez
 \begin{array}{cc} t & \\ \uparrow & \\ & \rightarrow x \end{array}                
\eez
\end{minipage}
\includegraphics[scale=.2]{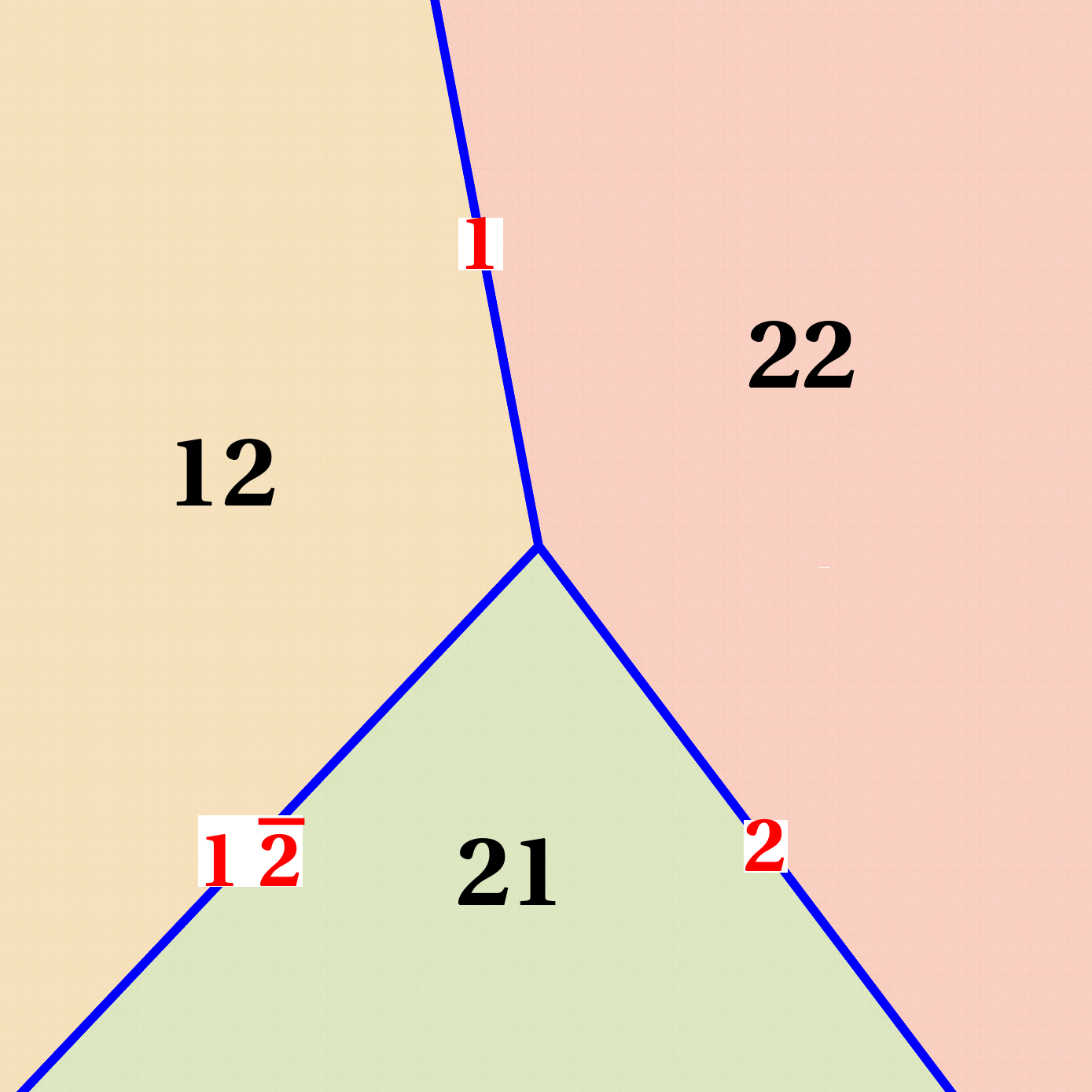} 
\hspace{.5cm}
\includegraphics[scale=.2]{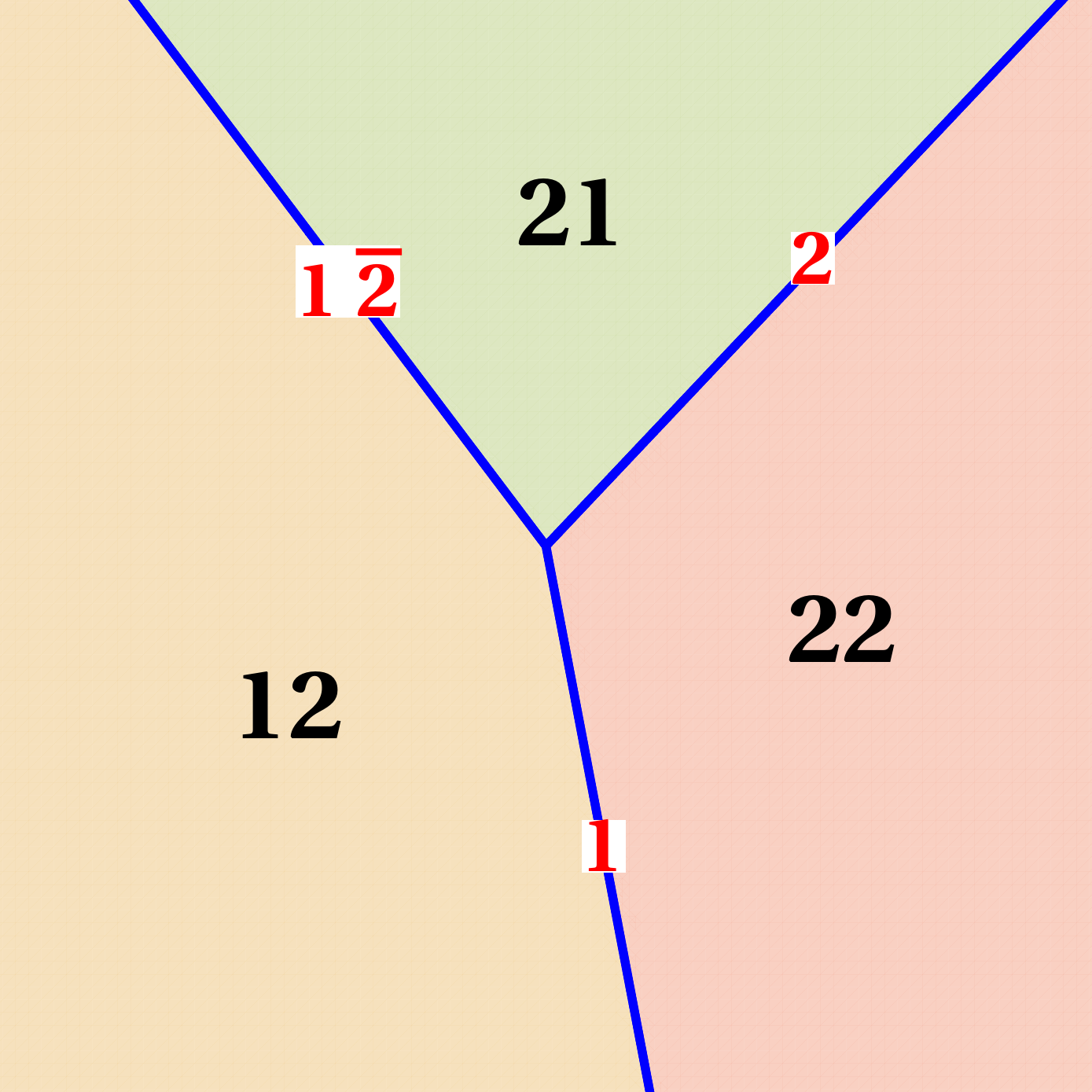}
\end{center}
\caption{Tropical limit graphs of degenerate 2-soliton solutions of the $m=2$ vector (i.e., $n=1$) 
Bsq$_K$ equation. Here we chose $K = (1,1)$, $\eta_1 = \eta_2 =1$, 
$\xi_1 = (1,0)^T$, $\xi_2 = (0,1)^T$, and the parameter values 
$(\lambda_1,\lambda_2) = (7/10,17/11)$ (so that $q_2 = q_1$), respectively $(\lambda_1,\lambda_2) = (7/10,3/31)$
(so that $p_2=p_1$).  
\label{fig:pure_2-soliton_degen} }
\end{figure}

In these cases the $R$-matrix reduces to
\bez
    R(\lambda,\frac{1-\lambda}{3 \lambda +1}) = \left( \begin{array}{cc} 
      0 & \frac{4 \lambda}{(1-\lambda)(3 \lambda+1)} \\[.5em]
      1 & \frac{(1+\lambda)(3 \lambda-1)}{(\lambda-1)(3 \lambda+1)} 
      \end{array} \right) \, ,
\eez
respectively
\bez
    R(\lambda,\frac{1+\lambda}{3 \lambda-1}) = \left( \begin{array}{cc} 
      \frac{(1-\lambda)(1 + 3 \lambda)}{4 \lambda} & 1 \\[.5em]
      \frac{(1+\lambda)(3 \lambda-1)}{4 \lambda} & 0 
      \end{array} \right) \, .
\eez

\subsubsection{$N$-soliton solutions and phase shifts}
If $I=(a_1,\ldots,a_N)$ and $i<j$, let
\bez
     I_{ij}(a,b) = (a_1,\ldots,a_{i-1},a,a_{i+1},\ldots,a_{j-1},b,a_{j+1},\ldots,a_N) \, .
\eez
The tropical limit graph of an $N$-soliton solution generically contains subgraphs describing 
2-soliton interactions, see Fig.~\ref{fig:2soliton_in_N}.
In the vicinity of such a local 2-soliton interaction, the solution is well approximated by only keeping 
the four relevant phases (since exponentials of the others are then negligible) in the function $\tau$. 
Hence  
\bez
    \tau \approx \tau^I_{ij} := \sum_{a,b=1,2} \tau_{I_{ij}(a,b)} \, , \qquad
    \phi \approx \phi^I_{ij} := \frac{1}{\tau^I_{ij}} \sum_{a,b=1,2} \phi_{I_{ij}(a,b)} \, \tau_{I_{ij}(a,b)} \, .
\eez
Also see, e.g., \cite{Koda10} (Section~5 therein) for such approximations in the case of KP solitons. 

\begin{figure} 
\begin{center}
\begin{minipage}{0.04\linewidth}
\vspace*{.15cm}
\bez
 \begin{array}{cc} t & \\ \uparrow & \\ & \rightarrow x \end{array}                
\eez
\end{minipage}
\includegraphics[scale=.2]{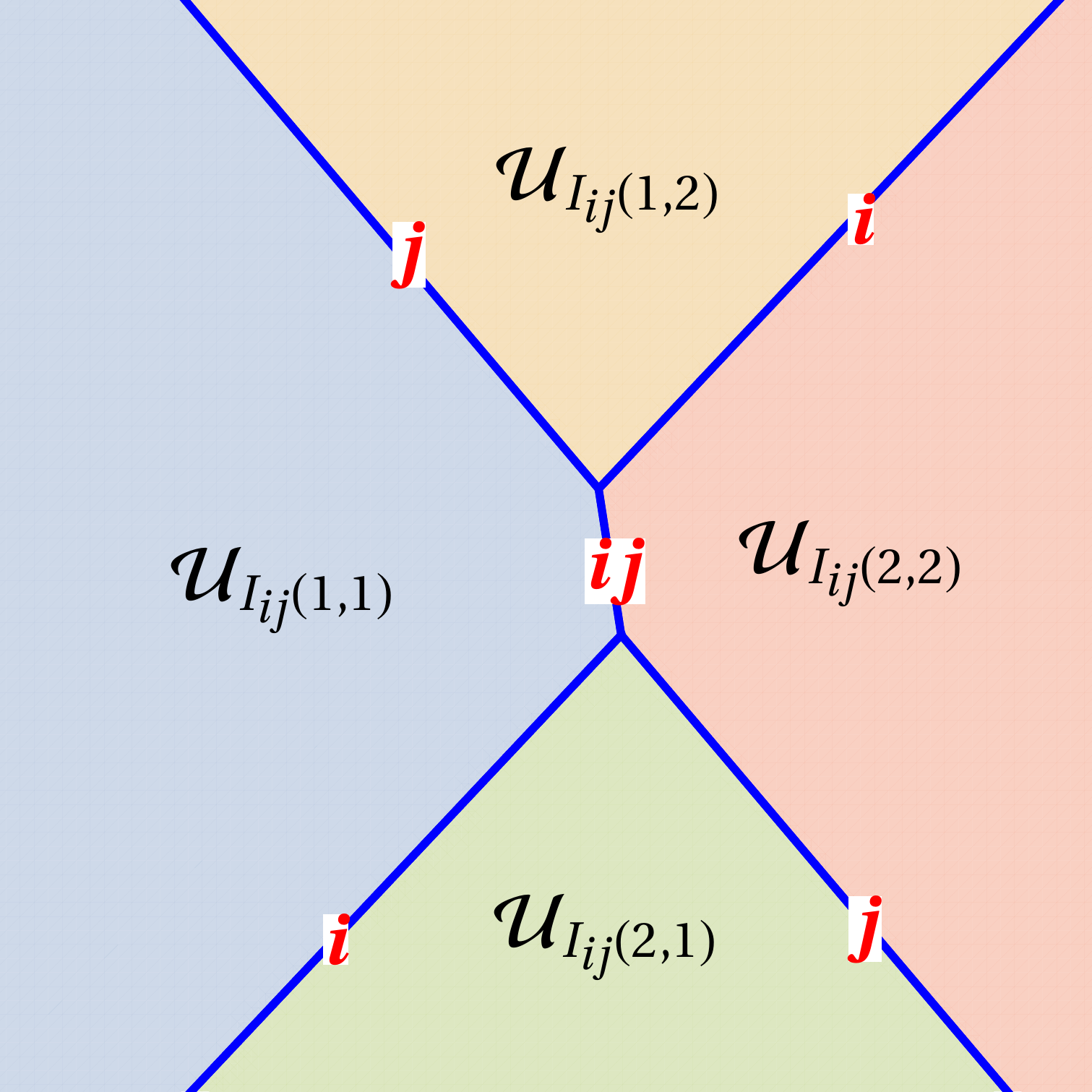} 
\end{center}
\caption{A 2-soliton part of an $N$-soliton tropical limit graph.
\label{fig:2soliton_in_N} }
\end{figure}

The two parallel line segments corresponding to the path of the $i$-th 
soliton are determined by
\bez
    x_{i,in} = -(p_i+q_i) \, t - (p_i-q_i)^{-1} \log\Big( \frac{\mu_{I_{ij}(1,1)}}{\mu_{I_{ij}(2,1)}} \Big) \, , \quad
    x_{i,out} = -(p_i+q_i) \, t - (p_i-q_i)^{-1} \log\Big( \frac{\mu_{I_{ij}(1,2)}}{\mu_{I_{ij}(2,2)}} \Big) .
\eez
Their shift along the $x$-axis, caused by the interaction with the $j$-th soliton, is
\bez
     \delta^I_{ij}x := x_{i,out} - x_{i,in} = (p_i-q_i)^{-1} \log( A_{I_{ij}} ) \, , \qquad
  A_{I_{ij}} = \frac{\mu_{I_{ij}(1,1)} \mu_{I_{ij}(2,2)}}{\mu_{I_{ij}(1,2)} \mu_{I_{ij}(2,1)}}      \, .
\eez
The parallel line segments of the $j$-th soliton are given by 
\bez
    x_{j,in} = -(p_j+q_j) \, t - (p_j-q_j)^{-1} \log\Big( \frac{\mu_{I_{ij}(2,1)}}{\mu_{I_{ij}(2,2)}} \Big) \, , \quad
    x_{j,out} = -(p_j+q_j) \, t - (p_j-q_j)^{-1} \log\Big( \frac{\mu_{I_{ij}(1,1)}}{\mu_{I_{ij}(1,2)}} \Big) .
\eez
Their shift along the $x$-axis, caused by the interaction with the $i$-th soliton, is 
\bez
    \delta^I_{ji}x := x_{j,out} - x_{j,in} = - (p_j-q_j)^{-1} \log( A_{I_{ij}} ) \, .
\eez
If $N=2$, we have $A_{I_{1,2}} = \alpha/(\kappa_{1,1} \kappa_{2,2})$ (also see \cite{Hiro+Ito83}
for the scalar case).

\subsubsection{3-soliton solutions}
\label{subsec:3solitons}
If $N=3$, we find
\bez
  \tau &=& \gamma \, e^{\vartheta_{1,1,1}} 
         + \alpha_{1,2} \, e^{\vartheta_{1,1,2}}
         + \alpha_{1,3} \, e^{\vartheta_{1,2,1}}
         + \alpha_{2,3} \, e^{\vartheta_{2,1,1}}  \\
     &&  + \kappa_{1,1} \, e^{\vartheta_{1,2,2}} 
         + \kappa_{2,2} \, e^{\vartheta_{2,1,2}} 
         + \kappa_{3,3} \, e^{\vartheta_{2,2,1}} 
         + e^{\vartheta_{2,2,2}} \, ,
\eez
where $\kappa_{ij}$, $i,j=1,2,3$, are defined as in (\ref{kappa_ij}), and  
\bez
     \alpha_{ij} &=& \kappa_{ii} \kappa_{jj} - \frac{(p_i-q_i)(p_j-q_j)}{(p_j-q_i)(p_i-q_j)}  \kappa_{ij} \kappa_{ji} \, , \\
     \gamma &=&  \kappa_{1,1} \kappa_{2,2} \kappa_{3,3} 
        + \frac{(p_1-q_1)(p_3-q_3)(p_2-q_2)}{(p_2-q_1)(p_3-q_2)(p_1-q_3)}  \kappa_{1,2} \kappa_{2,3} \kappa_{3,1} \\
     && + \frac{(p_1-q_1)(p_3-q_3)(p_2-q_2)}{(p_3-q_1)(p_1-q_2)(p_2-q_3)}  \kappa_{1,3} \kappa_{2,1} \kappa_{3,2} 
        - \frac{(p_1-q_1)(p_2-q_2)}{(p_2-q_1)(p_1-q_2)}  \kappa_{1,2} \kappa_{2,1} \kappa_{3,3}  \\
     && - \frac{(p_3-q_3)(p_2-q_2)}{(p_3-q_2)(p_2-q_3)} \kappa_{1,1} \kappa_{2,3} \kappa_{3,2}
        - \frac{(p_1-q_1)(p_3-q_3)}{(p_3-q_1)(p_1-q_3)}  \kappa_{1,3} \kappa_{2,2} \kappa_{3,1} \, .
\eez
Furthermore, if all coefficients of exponentials in $\tau$ are positive, we have the following tropical values of $\phi$,
\bez
 \phi_{1,1,1} &=& \frac{\kappa_{1,1} \kappa_{2,2} \kappa_{3,3}}{\gamma} \Big(
                 \tilde{\alpha}_{1,2} (p_3-q_3) \, \xi_3 \otimes \eta_3
               + \tilde{\alpha}_{1,3} (p_2-q_2) \, \xi_2 \otimes \eta_2
               + \tilde{\alpha}_{2,3} (p_1-q_1) \, \xi_1 \otimes \eta_1 \\
          &&   + (1-\tilde{\alpha}_{2,3}) \tilde{\alpha}_{1,2,3} (q_3-p_2) \frac{\xi_2\otimes \eta_3}{\kappa_{3,2}}
               + (1-\tilde{\alpha}_{2,3}) \tilde{\alpha}_{1,3,2} (q_2-p_3) \frac{\xi_3\otimes \eta_2}{\kappa_{2,3}} \\
          &&   + (1-\tilde{\alpha}_{1,3}) \tilde{\alpha}_{2,1,3} (q_3-p_1) \frac{\xi_1\otimes \eta_3}{\kappa_{3,1}}
               + (1-\tilde{\alpha}_{1,3}) \tilde{\alpha}_{2,3,1} (q_1-p_3) \frac{\xi_3\otimes \eta_1}{\kappa_{1,3}} \\
          &&   + (1-\tilde{\alpha}_{1,2}) \tilde{\alpha}_{3,1,2} (q_2-p_1) \frac{\xi_1\otimes \eta_2}{\kappa_{2,1}}
               + (1-\tilde{\alpha}_{1,2}) \tilde{\alpha}_{3,2,1} (q_1-p_2) \frac{\xi_2\otimes \eta_1}{\kappa_{1,2}} \Big) \, , \\
    \phi_{1,1,2} &=& \frac{1}{\tilde{\alpha}_{1,2}} \Big( 
                   (p_1-q_1) \frac{\xi_1\otimes \eta_1}{\kappa_{1,1}}
                 + (p_2-q_2) \frac{\xi_2\otimes \eta_2}{\kappa_{2,2}}
                 + (q_2-p_1)(1-\tilde{\alpha}_{1,2}) \frac{\xi_1\otimes \eta_2}{\kappa_{2,1}} \\
            &&   + (q_1-p_2)(1-\tilde{\alpha}_{1,2}) \frac{\xi_2\otimes \eta_1}{\kappa_{1,2}} \Big)   \, , \\
    \phi_{1,2,1} &=& \frac{1}{\tilde{\alpha}_{1,3}} \Big(
                   (p_1-q_1) \frac{\xi_1\otimes \eta_1}{\kappa_{1,1}}
                 + (p_3-q_3) \frac{\xi_3\otimes \eta_3}{\kappa_{3,3}}
                 + (q_3-p_1)(1-\tilde{\alpha}_{1,3}) \frac{\xi_1\otimes \eta_3}{\kappa_{3,1}} \\
            &&   + (q_1-p_3)((1-\tilde{\alpha}_{1,3}) \frac{\xi_3\otimes \eta_1}{\kappa_{1,3}} \Big)   \, , \\
    \phi_{2,1,1} &=& \frac{1}{\tilde{\alpha}_{2,3}} \Big(
                   (p_2-q_2) \frac{\xi_2\otimes \eta_2}{\kappa_{2,2}}
                 + (p_3-q_3) \frac{\xi_3\otimes \eta_3}{\kappa_{3,3}}
                 + (q_3-p_2)(1-\tilde{\alpha}_{2,3}) \frac{\xi_2 \otimes \eta_3}{\kappa_{3,2}} \\
            &&   + (q_2-p_3)(1-\tilde{\alpha}_{2,3}) \frac{\xi_3 \otimes \eta_2}{\kappa_{2,3}} \Big)   \, , \\
    \phi_{1,2,2} &=& (p_1-q_1) \frac{\xi_1\otimes \eta_1}{\kappa_{1,1}} \, , \qquad 
    \phi_{2,1,2} = (p_2-q_2) \frac{\xi_2\otimes \eta_2}{\kappa_{2,2}} \, , \\
    \phi_{2,2,1} &=& (p_3-q_3) \frac{\xi_3\otimes \eta_3}{\kappa_{3,3}} \, , \qquad
    \phi_{2,2,2} = 0 \, ,    
\eez
where
\bez
    \tilde{\alpha}_{ij} &=& 1 - \frac{(p_i-q_i)(p_j-q_j) \, \kappa_{ij} \kappa_{ji}}{(p_j-q_i)(p_i-q_j) \, \kappa_{ii} \kappa_{jj}} \, , \qquad
    \tilde{\alpha}_{kij} = 1 - \frac{(p_k-q_k) (p_j-q_i) \, \kappa_{ik} \kappa_{kj}}{(p_k-q_i) (p_j-q_k) \, \kappa_{kk} \kappa_{ij}} \, .
\eez

Examples of corresponding tropical limit graphs are shown in Fig.~\ref{fig:3-soliton}. Here we extended 
the phase expression (\ref{Bouss_phase}) by including the next hierarchy variable:
\bez
   \vartheta(P) = P \, x + P^2 \, t + P^4 \, s \, .
\eez
The first and the third graph correspond to a large negative, respectively large positive value 
of $s$. The sequences of 2-soliton interactions are according to the left, respectively right 
hand side of the Yang-Baxter equation. Since the polarizations along edges of a tropical 
limit graph do not depend on the variables ($x,t,s$), we conclude that the Yang-Baxter 
equation holds. 

\begin{figure} 
\begin{center}
\begin{minipage}{0.04\linewidth}
\vspace*{.15cm}
\bez
 \begin{array}{cc} t & \\ \uparrow & \\ & \rightarrow x \end{array}                
\eez
\end{minipage}
\includegraphics[scale=.25]{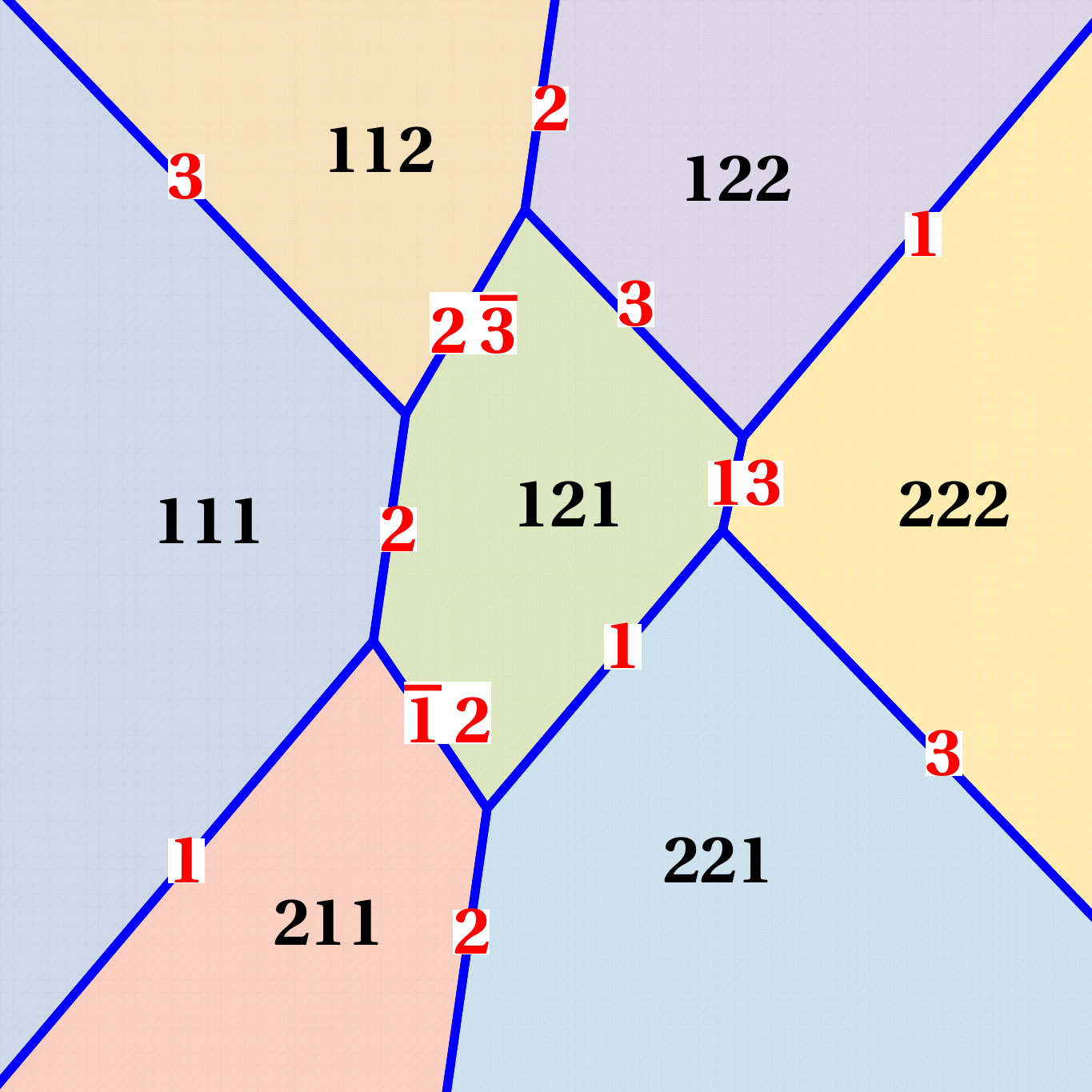} 
\hspace{.3cm}
\includegraphics[scale=.25]{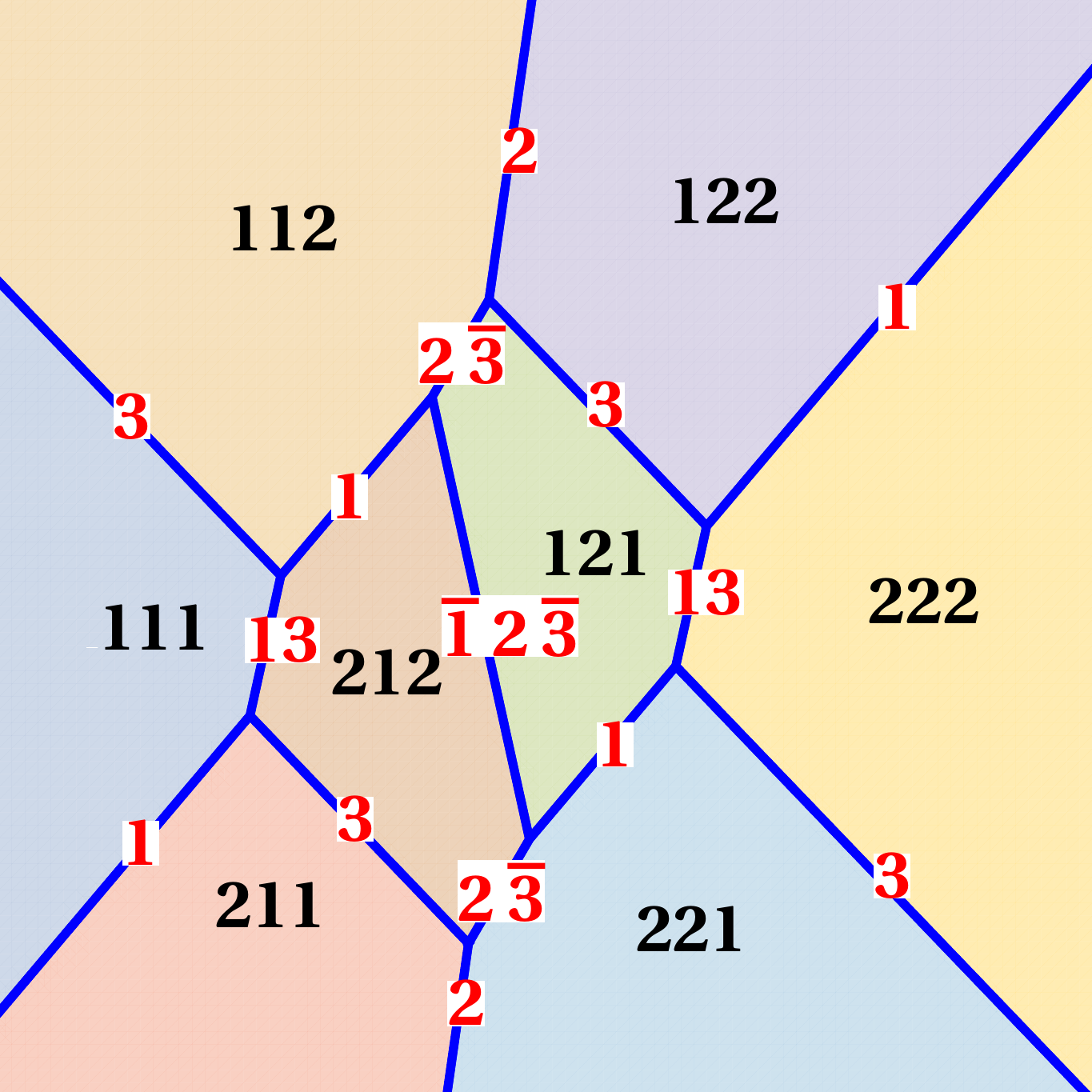} 
\hspace{.3cm}
\includegraphics[scale=.25]{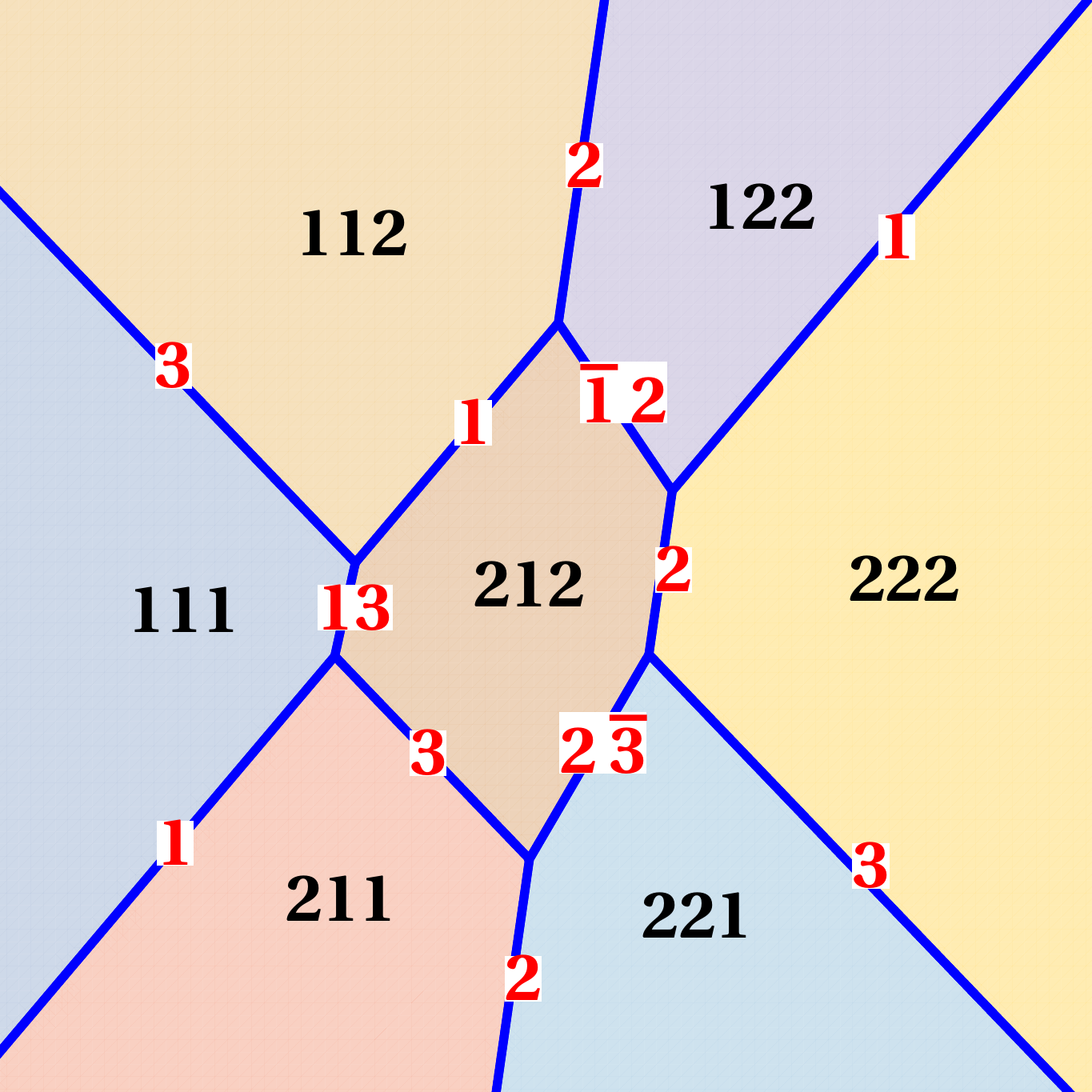} 
\end{center}
\caption{Tropical limit graphs of 3-soliton solutions of the $m=3$ vector ($n=1$) Bsq$_K$ 
equation with $K = (1,1,1)$, $\eta_i = 1$, $i=1,2,3$, 
$\xi_1 = (1,0,0)^T$, $\xi_2 = (0,1,0)^T$, $\xi_3 = (0,0,1)^T$, and $p_i$, $q_i$ are given by 
(\ref{Bouss_p,q__2-soliton}) with $(\lambda_1,\lambda_2,\lambda_3) = (1/6,1/2,4)$. 
The three graphs correspond to consecutive values of the next hierarchy variable, $s = -50,10,50$. 
According to the general asymptotic structure of pure soliton solutions, $\mathcal{U}_{1,2,1}$ 
and $\mathcal{U}_{2,1,2}$ are the only possible interior regions. 
\label{fig:3-soliton} }
\end{figure}

\subsection{Other soliton configurations}
Now we consider solutions involving three roots of the cubic equation. There are then two cases, either
\bez
    \theta = \theta_1 \, e^{\vartheta(P_1)} + \theta_2 \, e^{\vartheta(P_2)} \, , \qquad
    \chi = e^{-\vartheta(P_3)} \, \chi_3 \, , 
\eez
or
\bez
    \theta = \theta_3 \, e^{\vartheta(P_3)} \, , \qquad 
    \chi = e^{-\vartheta(P_1)} \, \chi_1 + e^{-\vartheta(P_2)} \, \chi_2 \, . 
\eez
The second, ``dual'' choice can be obtained from the first by applying the symmetry $\phi \mapsto \phi^T$, $t \mapsto -t$, 
and using $\chi_a \leftrightarrow \theta_a^T$, $P_a \mapsto -P_a^T$. 
In the following we use again the decomposition (\ref{chi,theta_decomp}), but with the rescaling 
$\xi_{ia} \mapsto (p_{i,3} - p_{ia}) \, \xi_{ia}$, $a=1,2$.

\begin{example}
\label{ex:2to1}
We consider the simplest case, $N=1$. Hence $i=1$ in $p_{ia}$ and $\xi_{ia}$. The corresponding 
index will now be suppressed. Then we have 
\bez
    \theta = (p_3 - p_1) \xi_1 \, e^{\vartheta(p_1)} + (p_3-p_2) \xi_2 \, e^{\vartheta(p_2)} \, , \qquad
    \chi = e^{-\vartheta(p_3)} \, \eta_3 \, .
\eez
Here $\xi_a$, $a=1,2$, are $m$-component column vectors, $\eta_3$ is an $n$-component row vector. Then
\bez
    \tau = e^{\vartheta(p_3)} \Omega 
         = e^{\vartheta(p_3)} + \eta_3 K \xi_1 \, e^{\vartheta(p_1)} + \eta_3 K \xi_2 \, e^{\vartheta(p_2)} 
\eez
and
\bez
    \phi = - \frac{1}{\tau} \Big( (p_3 - p_1) \xi_1 \otimes \eta_3 \, e^{\vartheta(p_1)} 
             + (p_3-p_2) \xi_2 \otimes \eta_3 \, e^{\vartheta(p_2)} \Big) \, . 
\eez
We set
\bez
    p_1 = - \frac{\sqrt{\beta } \left(1+6 \lambda-3 \lambda^2\right)}{1+3 \lambda^2} \, , \quad
    p_2 = - \frac{\sqrt{\beta } \left(1-6 \lambda-3 \lambda^2\right)}{1+3 \lambda^2} \, , \quad
    p_3 = \frac{2 \sqrt{\beta } \left(1-3 \lambda^2\right)}{1+3 \lambda^2} \, .
\eez
The solution describes the merging of two solitons into a single one, also see Fig.~\ref{fig:2to1}.  
If $\lambda \in \{ 0, \pm 1/3, \pm 1 \}$, two of the $p_i$ are equal and the solution reduces to a 1-soliton solution.
The dual case describes the splitting of a single soliton into two.
\begin{figure} 
\begin{center}
\includegraphics[scale=.3]{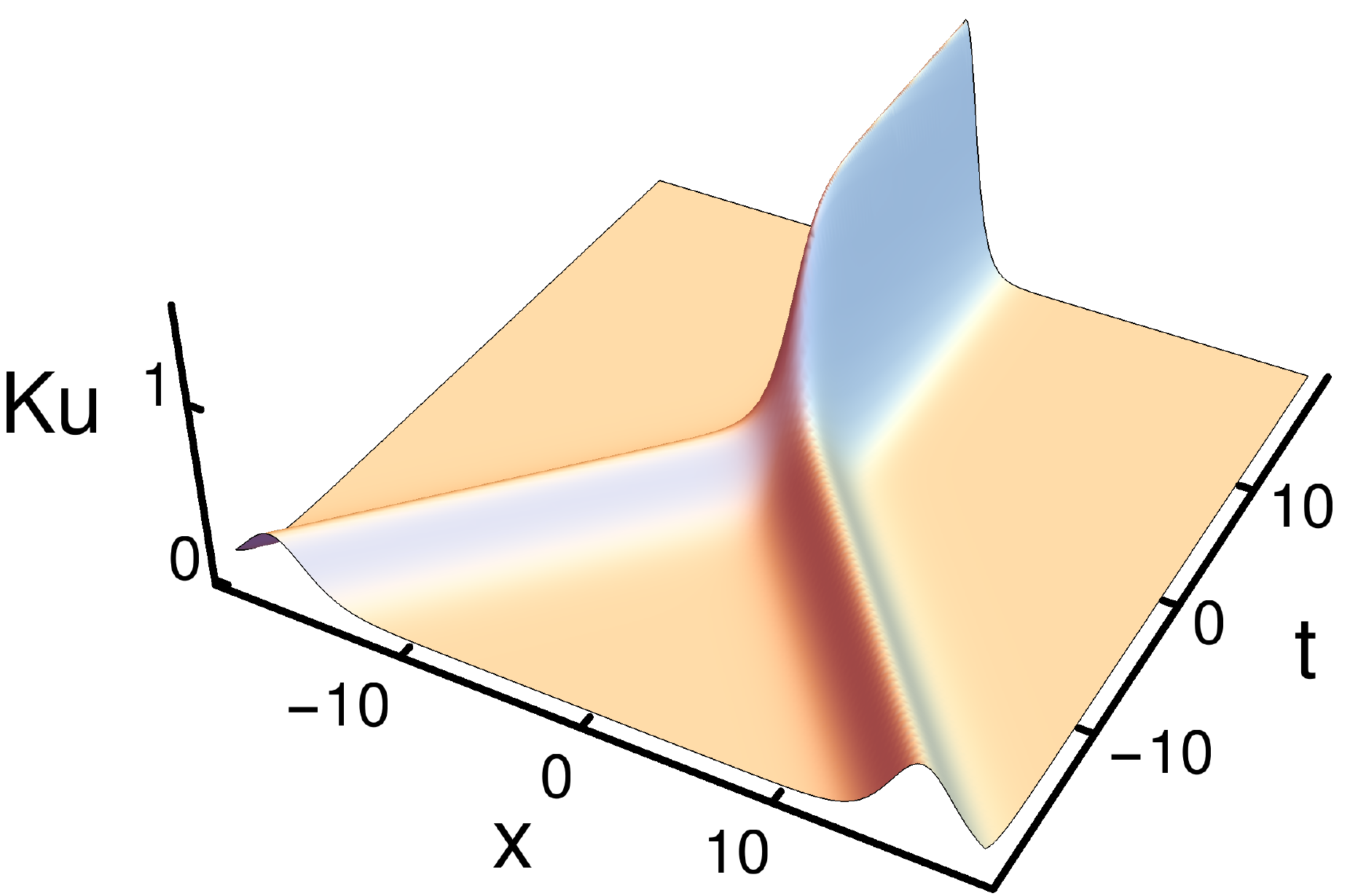} 
\hspace{.3cm}
\begin{minipage}{0.04\linewidth}
\vspace*{.15cm}
\bez
 \begin{array}{cc} t & \\ \uparrow & \\ & \rightarrow x \end{array}                
\eez
\end{minipage}
\includegraphics[scale=.18]{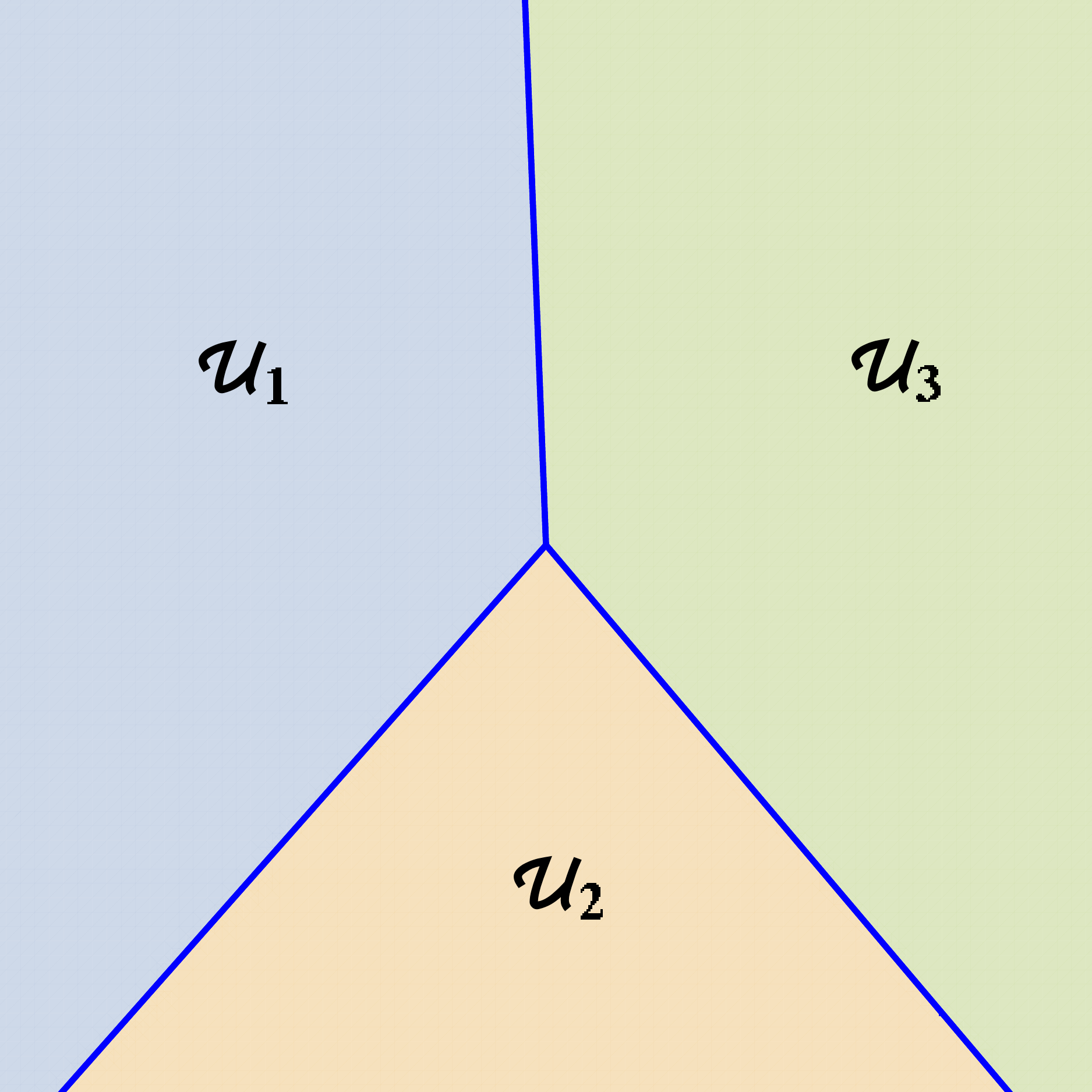} 
\end{center}
\caption{Merging of two $m=2$ vector ($n=1$) solitons into one, according to Example~\ref{ex:2to1}.  
Here we chose $K=(1,1)$, $\theta_1 = (1,0)^T$, $\theta_2 = (0,1)^T$, $\eta_3 =1$ and $\lambda =1/7$. 
\label{fig:2to1} }
\end{figure}
\end{example}

\begin{remark}
In contrast to the KdV reduction of KP, where a quadratic equation rules the game, Bsq$_K$ thus admits 
solutions with a tropical limit graph in space-time having the most elementary rooted binary tree shape: three 
edges meeting at a vertex. We may speculate 
that in higher Gelfand-Dickey reductions there are corresponding limits for the number of edges forming 
a rooted binary tree. 
\end{remark}

\begin{example}
\label{ex:superpos_trees}
For $N=2$ and $P_a = \mathrm{diag}(p_{1,a},p_{2,a})$, $a=1,2,3$, 
we obtain
\bez
     \phi = \frac{F}{\tau}
\eez
with 
\bez
   \tau &=& e^{\vartheta(p_{1,3})+\vartheta(p_{2,3})} \det(\Omega)  \\
        &=& \sum_{a,b=1}^2 \alpha_{ab} \, e^{\tilde{\vartheta}_{a b}} 
         + \kappa_{1,1,1} \, e^{\tilde{\vartheta}_{1,3}} 
         + \kappa_{1,1,2} \, e^{\tilde{\vartheta}_{2,3}} 
         + \kappa_{2,2,1} \, e^{\tilde{\vartheta}_{3,1}} 
         + \kappa_{2,2,2} \, e^{\tilde{\vartheta}_{3,2}} 
         + e^{\tilde{\vartheta}_{3,3}} 
\eez
and 
\bez
 F &\hspace{-.3cm}=& \hspace{-.3cm} (p_{1,3}-p_{1,1}) (p_{2,3}-p_{2,1}) \Big(
           \frac{\kappa_{1,1,1} \xi_{2,1} \otimes \eta_2}{p_{1,1}-p_{1,3}}
         + \frac{\kappa_{1,2,1} \xi_{1,1} \otimes \eta_2}{p_{1,3}-p_{2,1}}
         + \frac{\kappa_{2,1,1} \xi_{2,1} \otimes \eta_1}{p_{2,3}-p_{1,1}}
         + \frac{\kappa_{2,2,1} \xi_{1,1} \otimes \eta_1}{p_{2,1}-p_{2,3}} \Big) e^{\tilde{\vartheta}_{1,1}} \\
     &&\hspace{-.5cm} + (p_{1,3}-p_{1,2}) (p_{2,3}-p_{2,1}) \Big(
           \frac{\kappa_{1,1,2} \xi_{2,1} \otimes \eta_2}{p_{1,2}-p_{1,3}}
         + \frac{\kappa_{1,2,1} \xi_{1,2} \otimes \eta_2}{p_{1,3}-p_{2,1}}
         + \frac{\kappa_{2,1,2} \xi_{2,1} \otimes \eta_1}{p_{2,3}-p_{1,2}}
         + \frac{\kappa_{2,2,1} \xi_{1,2} \otimes \eta_1}{p_{2,1}-p_{2,3}} \Big) e^{\tilde{\vartheta}_{2,1}} \\
     &&\hspace{-.5cm} + (p_{1,3}-p_{1,1}) (p_{2,3}-p_{2,2}) \Big(
           \frac{\kappa_{1,1,1} \xi_{2,2} \otimes \eta_2}{p_{1,1}-p_{1,3}}
         + \frac{\kappa_{1,2,2} \xi_{1,1} \otimes \eta_2}{p_{1,3}-p_{2,2}}
         + \frac{\kappa_{2,1,1} \xi_{2,2} \otimes \eta_1}{p_{2,3}-p_{1,1}}
         + \frac{\kappa_{2,2,2} \xi_{1,1} \otimes \eta_1}{p_{2,2}-p_{2,3}} \Big) e^{\tilde{\vartheta}_{1,2}} \\
     &&\hspace{-.5cm} + (p_{1,3}-p_{1,2}) (p_{2,3}-p_{2,2}) \Big(
           \frac{\kappa_{1,1,2} \xi_{2,2} \otimes \eta_2}{p_{1,2}-p_{1,3}}
         + \frac{\kappa_{1,2,2} \xi_{1,2} \otimes \eta_2}{p_{1,3}-p_{2,2}}
         + \frac{\kappa_{2,1,2} \xi_{2,2} \otimes \eta_1}{p_{2,3}-p_{1,2}}
         + \frac{\kappa_{2,2,2} \xi_{1,2} \otimes \eta_1}{p_{2,2}-p_{2,3}} \Big) e^{\tilde{\vartheta}_{2,2}} \\
     &&\hspace{-.5cm} + (p_{1,1}-p_{1,3}) \, \xi_{1,1}\otimes \eta_1 \, e^{\tilde{\vartheta}_{1,3}}
       + (p_{1,2}-p_{1,3}) \, \xi_{1,2} \otimes \eta_1 \, e^{\tilde{\vartheta}_{2,3}}
       + (p_{2,1}-p_{2,3}) \, \xi_{2,1} \otimes \eta_2 \, e^{\tilde{\vartheta}_{3,1}}  \\
     &&\hspace{-.5cm} + (p_{2,2}-p_{2,3}) \, \xi_{2,2} \otimes \eta_2 \, e^{\tilde{\vartheta}_{3,2}}  \, .
\eez
Here we set
\bez
  \tilde{\vartheta}_{ab} &=& \vartheta(p_{1,a}) + \vartheta(p_{2,b}) \qquad \quad a,b = 1,2,3 \, , \\
  \kappa_{ija} &=& \eta_i \, K \, \xi_{j,a} \qquad \quad i,j =1,2 \, , \quad a = 1,2  \, , \\
  \alpha_{ab} &=& \kappa_{1,1,a} \kappa_{2,2,b} - \frac{(p_{1,a}-p_{1,3}) (p_{2,b}-p_{2,3})}
          {(p_{1,a}-p_{2,3})(p_{2,b}-p_{1,3})} \kappa_{2,1,a} \kappa_{1,2,b} \qquad a,b=1,2 \, , 
\eez
using (\ref{chi,theta_decomp}) with
\bez
  \chi_3 = \left( \begin{array}{c}
            \eta_1 \\ \eta_2 
                   \end{array} \right) .
\eez 
\end{example}

\begin{example}
\label{ex:trees_2x2}
According to Proposition~\ref{prop:vectorBoussN=2} below, in the scalar and vector case there is no regular 
solution in the class given by Example~\ref{ex:superpos_trees}, with all possible phases present 
in the expression for $\tau$. But corresponding regular solutions exist, for example, in the 
$2 \times 2$ matrix case.  
Fig.~\ref{fig:trees_2x2} shows two examples of tropical limit graphs. Here we chose 
$K= \mathrm{diag}(1,1)$ and $\lambda_1 = 7/10$, $\lambda_2=4$.
\begin{figure} 
\begin{center}
\hspace{.3cm}
\begin{minipage}{0.04\linewidth}
\vspace*{.15cm}
\bez
 \begin{array}{cc} t & \\ \uparrow & \\ & \rightarrow x \end{array}                
\eez
\end{minipage}
\includegraphics[scale=.226]{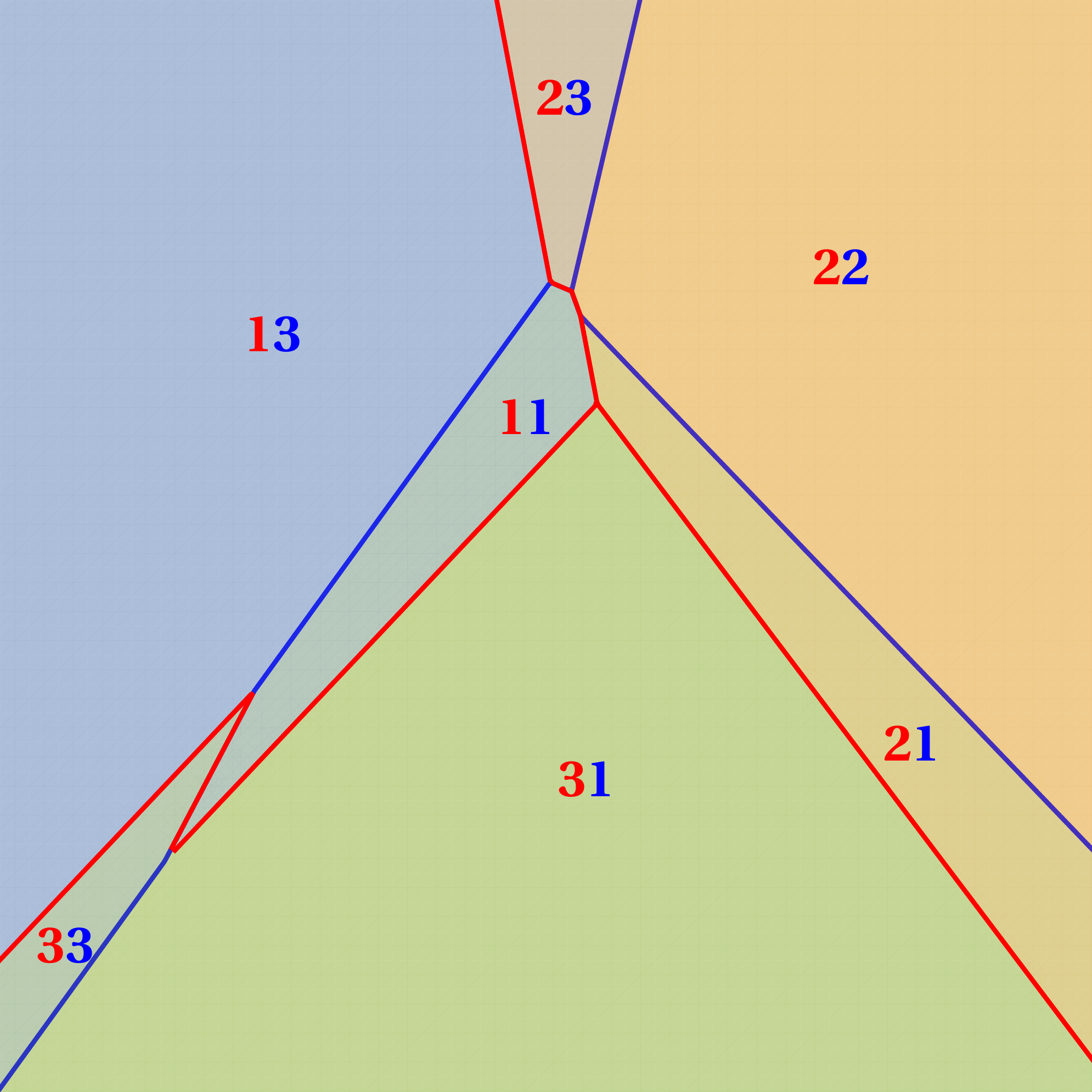} 
\hspace{.5cm} 
\includegraphics[scale=.2]{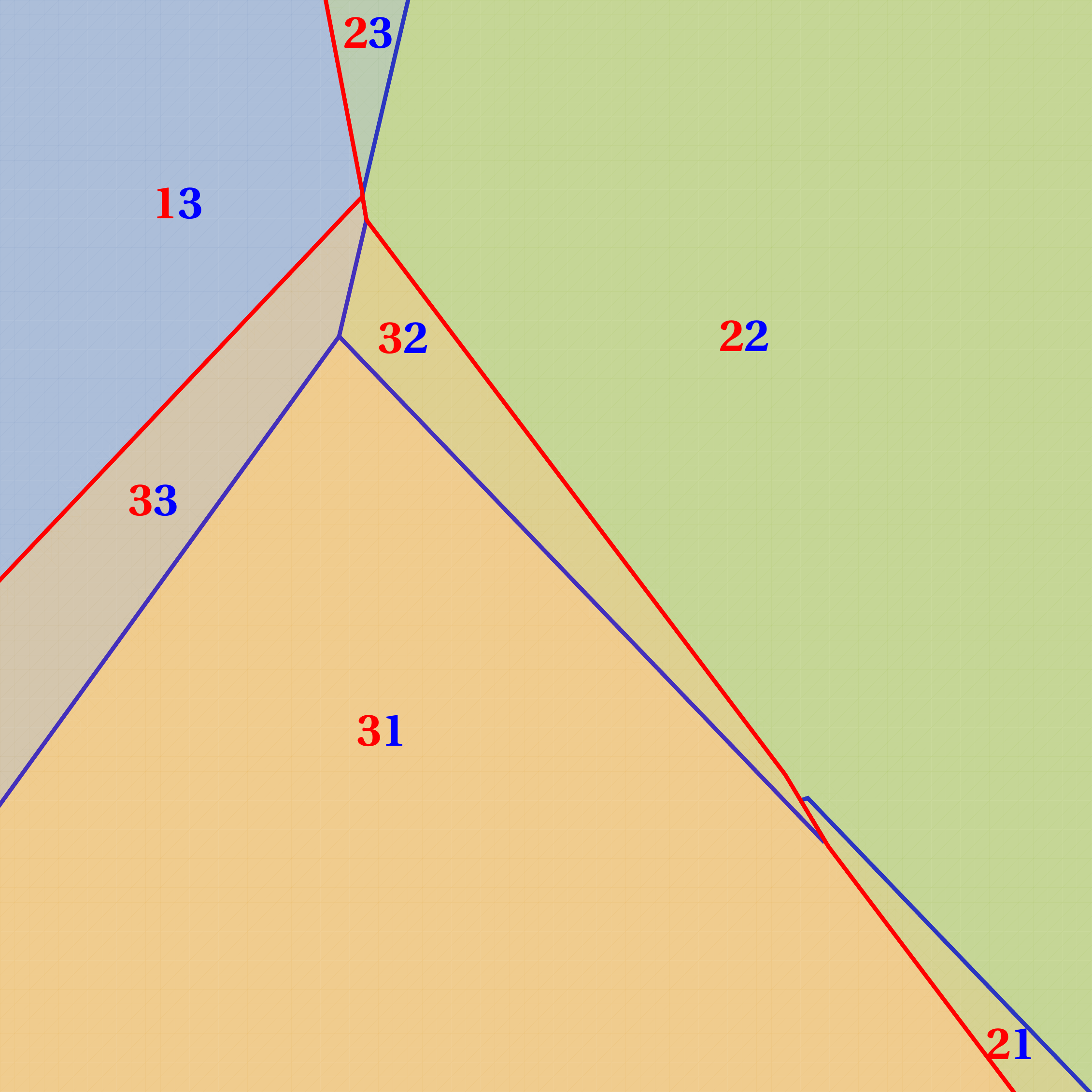} 
\end{center}
\caption{Tropical limit graphs of $2\times 2$ Boussinesq solutions  
according to Example~\ref{ex:trees_2x2}. 
In the first example, we set $\eta_1 = (10^2,0)$, $\eta_2 = (0,10^{-2})$,
$\xi_{1,1} = (1,1)^T$, $\xi_{1,2} = (1,0)^T$, $\xi_{2,1} = (0,1)^T$,  
$\xi_{2,2} = (1,2)^T$. In the second we chose $\eta_1 = (10^{-4},0)$, $\eta_2 = (0,10^4)$,
$\xi_{1,1} = (1,0)^T$, $\xi_{1,2} = (1,-2)^T$, $\xi_{2,1} = (2,1)^T$,  
$\xi_{2,2} = (1,1)^T$.
Only seven of the nine phases appearing in the expression for $\tau$ are visible. 
\label{fig:trees_2x2} }
\end{figure}
\end{example}

\subsubsection{Solutions of the vector Boussinesq equation}
\label{subsec:vBsq}
In this subsection we consider the case of the $m$-component vector ($n=1$) Boussinesq equation. This includes 
the scalar Boussinesq case ($m=1$). 

\paragraph{$\boldsymbol{N=2}$.} 

\begin{proposition}
\label{prop:vectorBoussN=2}
For the vector Boussinesq equation, and with real parameters, it is \emph{not} possible that all 
coefficients in the expression for $\tau$ in the $N=2$ case (Example~\ref{ex:superpos_trees}) 
are positive. 
\end{proposition}
\begin{proof}
In the vector case, where $\eta_i$ are scalars, we have $\kappa_{2,1,a} \kappa_{1,2,b} = \kappa_{1,1,a} \kappa_{2,2,b}$, hence
\be
  \alpha_{ab} = \kappa_{1,1,a} \, \kappa_{2,2,b} \, \tilde{\alpha}_{ab}  \, , \qquad
  \tilde{\alpha}_{ab} := \frac{(p_{2,3}-p_{1,3}) (p_{1,a}-p_{2,b})}{(p_{1,a}-p_{2,3})(p_{2,b}-p_{1,3})} \, .
     \label{vector_talpha}
\ee
Let us assume $\kappa_{iia} >0$, $i=1,2$, $a=1,2$, and $\alpha_{ab} >0$, $a,b=1,2$. This implies that, 
for $a=1,2$ and $b=1,2$, either 
\bez
    (p_{2,3}-p_{1,3}) (p_{1,a}-p_{2,b}) > 0 \quad \mbox{and} \quad (p_{1,a}-p_{2,3})(p_{2,b}-p_{1,3}) > 0 \, ,
\eez
or
\bez
    (p_{2,3}-p_{1,3}) (p_{1,a}-p_{2,b}) < 0 \quad \mbox{and} \quad (p_{1,a}-p_{2,3})(p_{2,b}-p_{1,3}) < 0 
\eez
holds. Next we observe that, for $i=1,2$, the solutions (\ref{roots}) of the cubic equation satisfy one of the following 
sets of inequalities (also see Fig.~\ref{fig:roots}). 
\begin{enumerate}
\item $p_{i,3} \leq -\sqrt{\beta} \leq p_{i,2} \leq \sqrt{\beta} < p_{i,1}$. 
\item $p_{i,2} \leq -\sqrt{\beta} \leq p_{i,3} \leq \sqrt{\beta} \leq p_{i,1}$.
\item $p_{i,2} \leq -\sqrt{\beta} \leq p_{i,1} \leq \sqrt{\beta} \leq p_{i,3}$. 
\end{enumerate}
This leaves us with nine cases.
In all of them, an analysis of the inequalities leads to a contradiction. 
\end{proof}

As a consequence of the above proposition, if all possible terms are 
present in the $N=2$ function $\tau$ (Example~\ref{ex:superpos_trees}), in the 
vector Boussinesq case the solution has a singularity.
Regular solutions are then only possible if at least one of the coefficients of exponentials 
in $\tau$ vanishes. In order to achieve this, we have to choose special values of some parameters, or special 
relations between parameters. 

To see what happens if one of the coefficients of exponentials in $\tau$ vanishes, let us consider the 
case where $K \xi_{1,2}=0$ (other cases lead to the same conclusions). Then the $N=2$ function $\tau$ reduces to
\be
 \tau &=& \alpha_{1,1} \, e^{\tilde{\vartheta}_{1,1}}
         + \alpha_{1,2} \, e^{\tilde{\vartheta}_{1,2}}  
         + \kappa_{1,1,1} \, e^{\tilde{\vartheta}_{1,3}} 
         + \kappa_{2,2,1} \, e^{\tilde{\vartheta}_{3,1}} 
         + \kappa_{2,2,2} \, e^{\tilde{\vartheta}_{3,2}} 
         + e^{\tilde{\vartheta}_{3,3}}  \nonumber \\
       &=& e^{\vartheta(p_{1,3})} \Big( e^{\vartheta(p_{2,3})} + K \xi_{2,1} \, e^{\vartheta(p_{2,1})} 
           + K \xi_{2,2} \, e^{\vartheta(p_{2,2})} \Big) \nonumber  \\
       &&  + K \xi_{1,1} \, e^{\vartheta(p_{1,1})} \Big( e^{\vartheta(p_{2,3})} 
             + \tilde{\alpha}_{1,1} \, K \xi_{2,1} \, e^{\vartheta(p_{2,1})} 
             + \tilde{\alpha}_{1,2} \, K \xi_{2,2} \, e^{\vartheta(p_{2,2})} \Big) , \label{tree_N=2_degen}
\ee
where we set $\eta_i =1$, which is no restriction of generality in the case under consideration, and 
$\tilde{\alpha}_{1,1}, \tilde{\alpha}_{1,2}$ are defined in (\ref{vector_talpha}). 
In the last expression in (\ref{tree_N=2_degen}), the part in the first brackets corresponds to an $N=1$ 
soliton configuration, cf. Example~\ref{ex:2to1}. 
If the part in the other brackets had the same coefficients, $\tau$ would factorize and we would obtain a 
tropical limit graph which is a superimposition of that of the $N=1$ Y-shaped solution and a single soliton. 
Different coefficients lead to a deformation, introducing phase shifts. 
Fig.~\ref{fig:tree_and_1soliton} shows an example. 

The condition $K \xi_{1,2}=0$ eliminates some phases from the function $\tau$, i.e., we have $\mu_{ab} =0$ 
for some $a,b$. If we choose $\xi_{1,2} = (1,-1)^T$, keeping otherwise the data specified in Fig.~\ref{fig:tree_and_1soliton}, 
then we still have $K \xi_{1,2}=0$, but some $M_{ab}$ is different from zero, although $\mu_{ab} =0$.    

\begin{figure} 
\begin{center}
\begin{minipage}{0.04\linewidth}
\vspace*{.15cm}
\bez
 \begin{array}{cc} t & \\ \uparrow & \\ & \rightarrow x \end{array}                
\eez
\end{minipage}
\includegraphics[scale=.2]{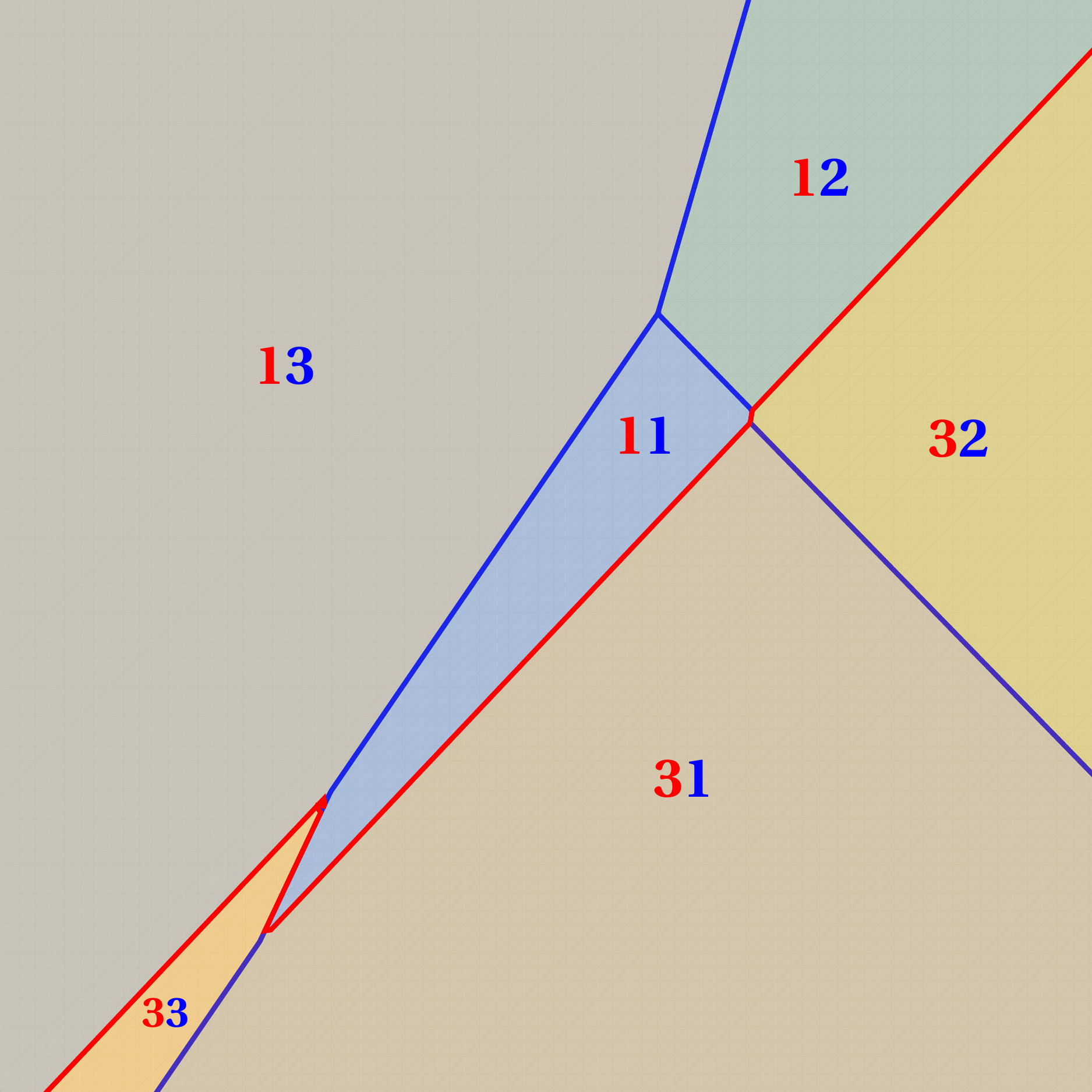} 
\hspace{1cm}
\includegraphics[scale=.2]{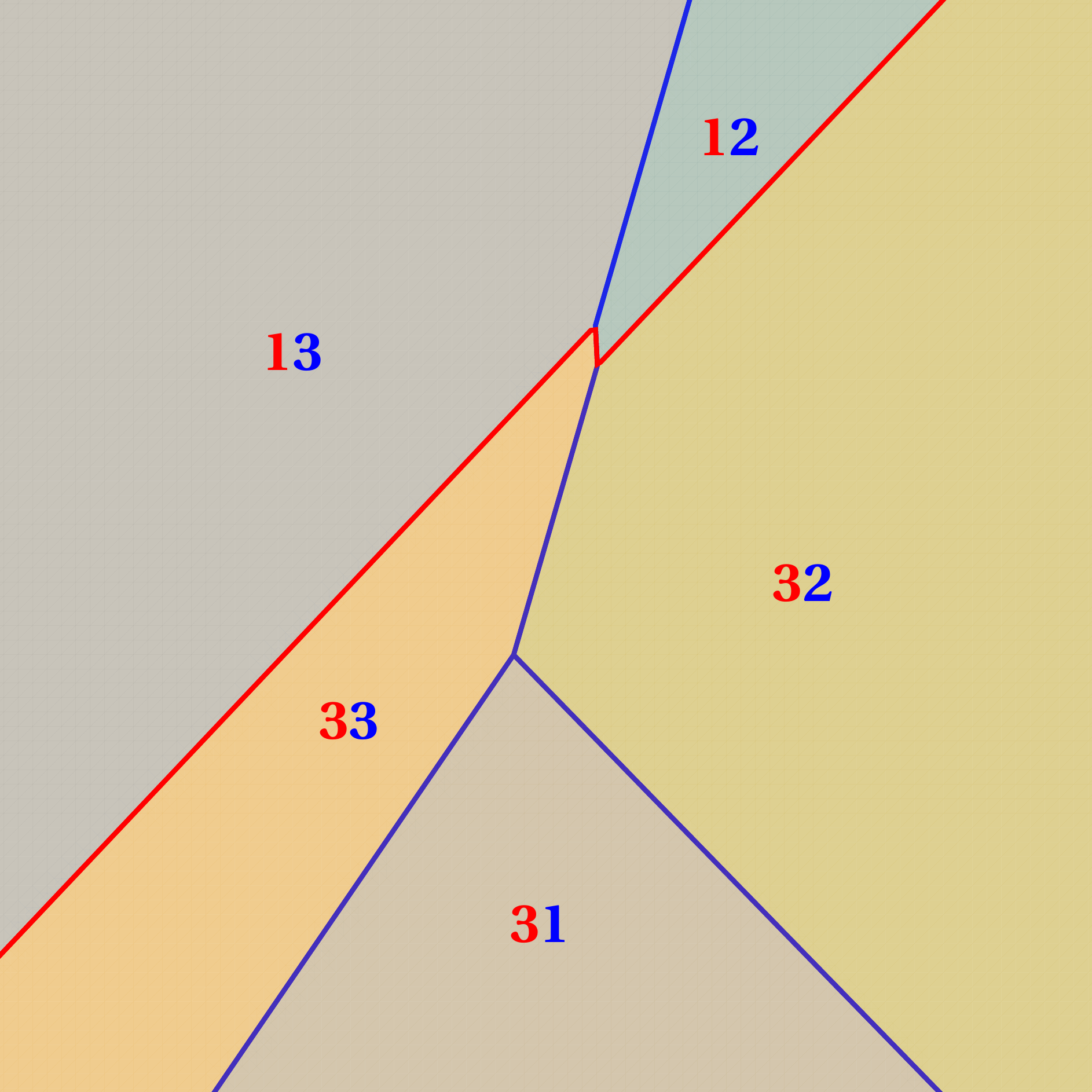} 
\end{center}
\caption{Tropical limit graph of a solution of the $m=2$ vector ($n=1$) Boussinesq equation from the class with $\tau$ 
given by (\ref{tree_N=2_degen}). 
The solution represents a (nonlinear) superposition of a branching soliton (marked blue) and a single soliton (marked red). 
Here we chose $K=(1,1)$, $\lambda_1 = 7/10$, $\lambda_2 = 5$, 
$\xi_{1,1} = (1,0)^T$, $\xi_{1,2} = (0,0)^T$, $\xi_{2,1} = (1,1)^T$, $\xi_{2,2} = (2,1)^T$ and 
$\eta_1 = 10^2$, $\eta_2 = 10^{-2}$, respectively $\eta_1 = 10^{-3}$, $\eta_2 = 10^{-3}$. 
The same graphs are obtained from a solution of the scalar Boussinesq equation, using $K=1$ and 
$\xi_{1,1} = 1$, $\xi_{1,2} =0$, $\xi_{2,1} = 2$, $\xi_{2,2} = 3$. 
\label{fig:tree_and_1soliton} }
\end{figure}

The last proposition tells us that, however small (with respect to a suitable norm) we choose a neighborhood 
of such a regular solution, in the set of solutions, it contains singular solutions. 

An $N>2$ solution locally consists approximately of $N=2$ solutions. We should thus expect the above proposition to 
extend to $N>2$. But we will not attempt to provide a rigorous proof here. In the following we show that 
solutions with $N=3$ exist, where the singularities only appear in a compact region of space-time.

\paragraph{$\boldsymbol{N=3}$.} We set $P_a = \mathrm{diag}(p_a, q_a, r_a)$, $a=1,2,3$, and obtain
\be
     \tau = \sum_{a,b,c=1}^3 \mu_{abc} \, e^{\tilde{\vartheta}_{abc}}   \label{tree_N=3}
\ee
with $\tilde{\vartheta}_{abc} = \vartheta(p_a) + \vartheta(q_a) + \vartheta(r_a)$ and 
\bez
   \mu_{abc} = \frac{(p_3-q_3) (p_3-r_3) (q_3-r_3) (p_a-q_b) 
  (p_a-r_c) (q_b-r_c)}{(p_a-q_3) (p_a-r_3) (q_b-p_3) 
   (q_b-r_3) (r_c-p_3) (r_c-q_3)}
   \frac{\eta_1 \, \eta_2 \, \eta_3 \, \kappa_{1,a} \, \kappa_{2,b} \, \kappa_{3,c}}
   {(-\eta_1 \kappa_{1,a}){}^{\delta_a^3} (-\eta_2 \kappa_{2,b}){}^{\delta_b^3} 
   (-\eta_3 \kappa_{3,c}){}^{\delta_c^3}} \, ,
\eez
where $\delta^i_a$ is the Kronecker symbol. Furthermore, 
\bez
  \Phi_{abc} &=& \frac{(p_a-p_3) (p_a-q_3) (p_a-r_3)}{(p_a-q_b) 
    (p_a-r_c)} \frac{\xi_{1,a}}{\kappa_{1,a}}
    + \frac{(q_b-q_3) (q_b-p_3) (q_b-r_3)}{(q_b-p_a) (q_b-r_c)} \frac{\xi_{2,b}}{\kappa_{2,b}}\\
   && + \frac{(r_c-r_3) (r_c-p_3) (r_c-q_3))}{(r_c-p_a) (r_c-q_b)} \frac{\xi_{3,c}}{\kappa_{3,c}} \, .
\eez

If $\tau$ has also \emph{negative} summands, strictly its tropical limit is not defined. Notwithstanding this, 
we can determine (and plot) the regions where the logarithm of the \emph{absolute value} of a summand 
of $\tau$ dominates. On a boundary segment between such (generalized) dominating phase regions, where 
the corresponding summand in $\tau$ is positive for one and negative for the other, the 
soliton solution is singular. This is so because close to such a boundary segment contributions from 
all other summands are exponentially suppressed, hence negligible. Furthermore, singularities can only 
appear at such a boundary segment, since everywhere else either a single summand of $\tau$ dominates and all others 
are negligible, or we have a boundary segment along which two summands with the same sign have equal 
values and all other summands are negligible. The latter two cases exclude a singularity.

Fig.~\ref{fig:N=3_singular} 
shows an example of such a modified tropical limit graph. The white regions are dominated by phases having 
negative contributions to the $\tau$-function. In these regions the solution is still regular, but on 
the interface between a ``positive'' and a ``negative'' phase region (plotted in red), and only there, 
the solution becomes singular. A similar solution of the scalar Boussinesq equation appeared in 
\cite{Rasi+Schi17} (see Figure 8 therein).

\begin{figure} 
\begin{center}
\begin{minipage}{0.04\linewidth}
\vspace*{.15cm}
\bez
 \begin{array}{cc} t & \\ \uparrow & \\ & \rightarrow x \end{array}                
\eez
\end{minipage}
\includegraphics[scale=.2]{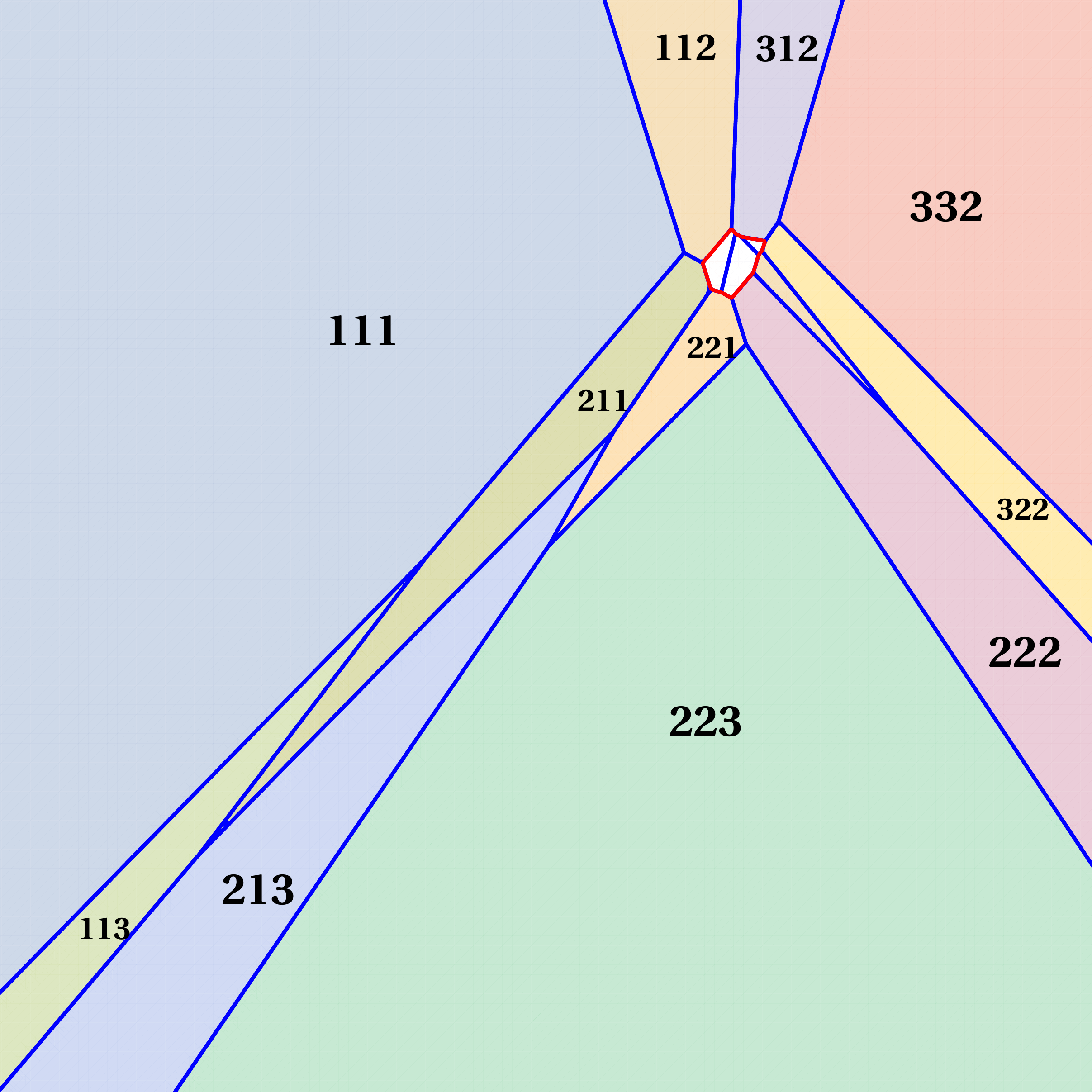} 
\hspace{1cm}
\includegraphics[scale=.2]{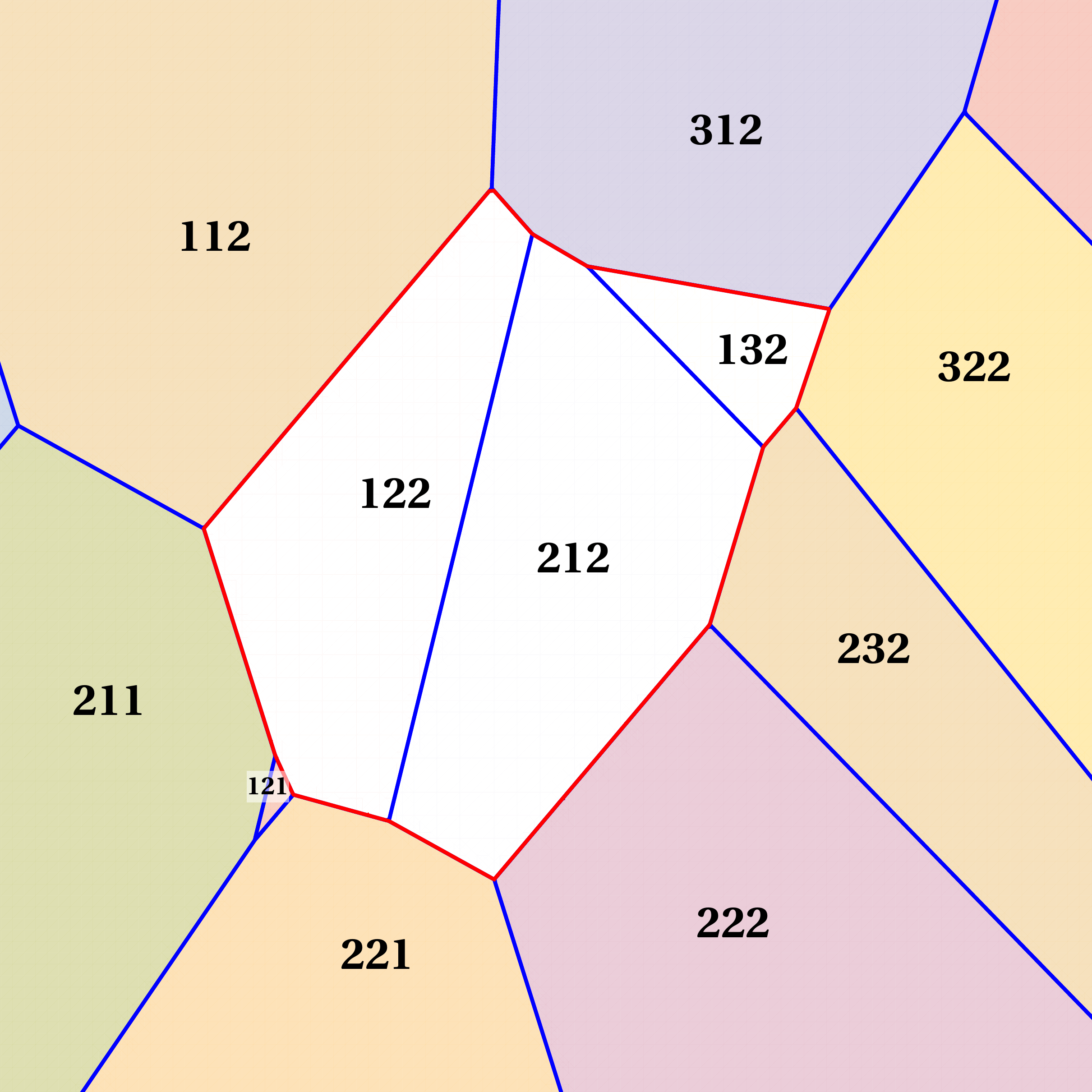} 
\end{center}
\caption{Modified tropical limit graph of a solution of the $m=2$ vector ($n=1$) Boussinesq equation 
from the class with $\tau$ given by (\ref{tree_N=3}). The second plot zooms into the region 
where most of the phases meet. 
There are 6 incoming and 3 outgoing solitons. The solution is singular only on the (red) 
boundary of the white region. 
Here we chose $K=(1,1)$, $\lambda_1 = 1/6$, $\lambda_2 = 1/4$, 
$\xi_{1,1} = (-1/2,0)^T$, $\xi_{2,1} = (1,-1/2)^T$, $\xi_{3,1} = (-3/2,-3/2)^T$,
$\xi_{1,2} = (0,-1)^T$, $\xi_{2,2} = (-3/2,-1)^T$, $\xi_{3,2} = (3/2,-3)^T$, and $\eta_i =1$ for $i=1,2,3$. 
The phases $\tilde{\vartheta}_{3,1,1}$ and $\tilde{\vartheta}_{3,1,3}$ are nowhere dominant, hence do not 
appear in the tropical limit plot. 
\label{fig:N=3_singular} }
\end{figure}

\section{The case $C' \neq -C$}
\label{sec:C'_noteq_-C}
We can use the transformations in Remark~\ref{rem:BDT_transf} to achieve that both, 
$C$ and $C'$, have Jordan normal form. Then (\ref{Omega_0-Sylv_eq}) generically implies $\Omega_0 =0$. 
This will be assumed in the following. Choosing diagonal matrices
\bez
  &&  P_a = \mathrm{diag}(p_{1,a}, \ldots,p_{N,a}) \, , \qquad
    Q_a = \mathrm{diag}(q_{1,a}, \ldots,q_{N,a}) \, , \\
  &&  C = \mathrm{diag}(c_1,\ldots,c_N) \, , \qquad
    C' = \mathrm{diag}(c'_1,\ldots,c'_N)  \, , 
\eez
(\ref{P_sol_cubic_eq}) and (\ref{Q_sol_cubic_eq}) require
\bez
    p_{ia}^3 = 3 \beta p_{ia} + c_i \, , \qquad
    q_{ia}^3 = 3 \beta q_{ia} - c'_i \, .  
\eez
For each $i=1,\ldots,N$, we represent the roots $p_{ia}$, $a=1,2,3$, of the first cubic equation 
as in (\ref{roots}), using a parameter $\lambda_i$. In the same way we represent the roots 
$q_{ia}$, $a=1,2,3$, of the second cubic equation using a parameter $\nu_i$.  
We assume that $p_{ia} \neq p_{jb}$ and $q_{ia} \neq q_{jb}$ for $i \neq j$ or $a \neq b$. 
Now we have
\bez
    \theta = \sum_{a=1}^3 \theta_a \, e^{\vartheta(P_a)} \, , \qquad
    \chi = \sum_{a=1}^3 e^{\vartheta(-Q_a)} \, \chi_a \, .
\eez
Assuming $q_{ib} \neq p_{ja}$ for all combinations of indices, 
from (\ref{Omega_ansatz}) and (\ref{Sylvester}) we obtain
\bez
  \Omega_{ij} = \sum_{a,b=1}^3 \frac{\eta_{ib} K \xi_{ja}}{ q_{ib}-p_{ja} } 
        \, e^{\vartheta(p_{ja})-\vartheta(q_{ib})} \qquad \quad i,j=1,\ldots,N  \, .
\eez

The Bsq$_K$ solution is given by
\bez
         \phi = \frac{F}{\tau} \, , 
\eez
where now
\bez
  && \tau = \det(\Omega) = \sum_{I,J \in \{1,2,3\}^N} \tau_{IJ} \, , \qquad
     \tau_{IJ} = \mu_{IJ} \, e^{\vartheta_{IJ}} \, , \qquad
     \vartheta_{IJ} = \sum_{i=1}^N \Big( \vartheta(p_{i a_i}) - \vartheta(p_{i b_i}) \Big) \, , \\
  && F = - \theta \, \mathrm{adj}(\Omega) \, \chi =:  \sum_{I,J \in \{1,2,3\}^N} \phi_{IJ} \, \tau_{IJ} \, , 
\eez
with constants $\mu_{IJ}$ and constant matrices $\phi_{IJ}$, of a certain structure. 
Here $\mathrm{adj}(\Omega)$ denotes 
the adjugate of the matrix $\Omega$, and we wrote 
$I=(a_1,\ldots,a_N)$, $J=(b_1,\ldots,b_N)$ in the expression for $\vartheta_{IJ}$. If a phase 
$\vartheta_{IJ}$ is present in $\tau$, and if the tropical limit exists, then the tropical 
value of $\phi$ in the corresponding dominating phase region is $\phi_{IJ}$. 
We decompose $\theta_a$ and $\chi_a$ as in (\ref{chi,theta_decomp}).

In the following we restrict our considerations to the case $N=1$.
Then $\Omega$ is a scalar, consisting of at most nine summands. We obtain $\phi = F/\tau$ with 
\bez
   && \tau = \Omega = \sum_{a,b=1}^3 \frac{\kappa_{ba}}{ q_b-p_a } 
        \, e^{\vartheta(p_a)-\vartheta(q_b)} \, , \qquad
      \kappa_{ba} = \chi_b K \theta_a    \, , \\
   && F = - \sum_{a,b=1}^3 \theta_a \otimes \chi_b \, e^{\vartheta(p_a)-\vartheta(q_b)} \, .
\eez
Hence
\bez
    \phi_{ab} = \frac{p_a - q_b}{\kappa_{ba}} \, \theta_a \otimes \chi_b \, .
\eez

In the scalar and vector Boussinesq case, 
one can show that there is no regular solution of the above form, with real parameters and with \emph{all} 
coefficients of exponentials in $\tau$ different from zero. This means that, however small (using 
a suitable norm) we choose a neighborhood of such a regular solution, in the set of solutions, 
it contains singular solutions.  
The first statement is not true in the matrix case. 

\begin{example}
\label{ex:C'=!-C_3x3}
Let $m=n=3$, $K = \mathrm{diag}(1,1,1)$, 
\bez
   && \chi_1 = \left( \begin{array}{ccc} 1 & 0 & 0 \end{array} \right) \, , \quad
    \chi_2 = \left( \begin{array}{ccc} 0 & 1 & 0 \end{array} \right) \, , \quad
    \chi_3 = \left( \begin{array}{ccc} 0 & 0 & 1 \end{array} \right) \, , \\
   && \theta_1 = \left( \begin{array}{r} 1 \\ 1 \\ 1 \end{array} \right) \, , \quad
    \theta_2 = \left( \begin{array}{r} -2 \\ 1 \\ -1 \end{array} \right) \, , \quad
    \theta_3 = \left( \begin{array}{r} -1 \\ 1 \\-1 \end{array} \right) \, , 
\eez
and $\lambda_1 = 1/5$, $\nu_1 = 3/2$. 
The function $\tau$, given above, is then a positive linear combination 
of nine independent exponentials, which is the maximal number. All nine phases, appearing in the 
corresponding function $\tau$, are visible in the tropical limit graph, see Fig.~\ref{fig:3x3}.
The three interior phase regions are hardly expected from the plot on the left hand side and reveal 
a complicated interaction pattern. 
Fig.~\ref{fig:3x3_2} shows another, though simpler example, where the relation is evident.  
\begin{figure} 
\begin{center}
\includegraphics[scale=.3]{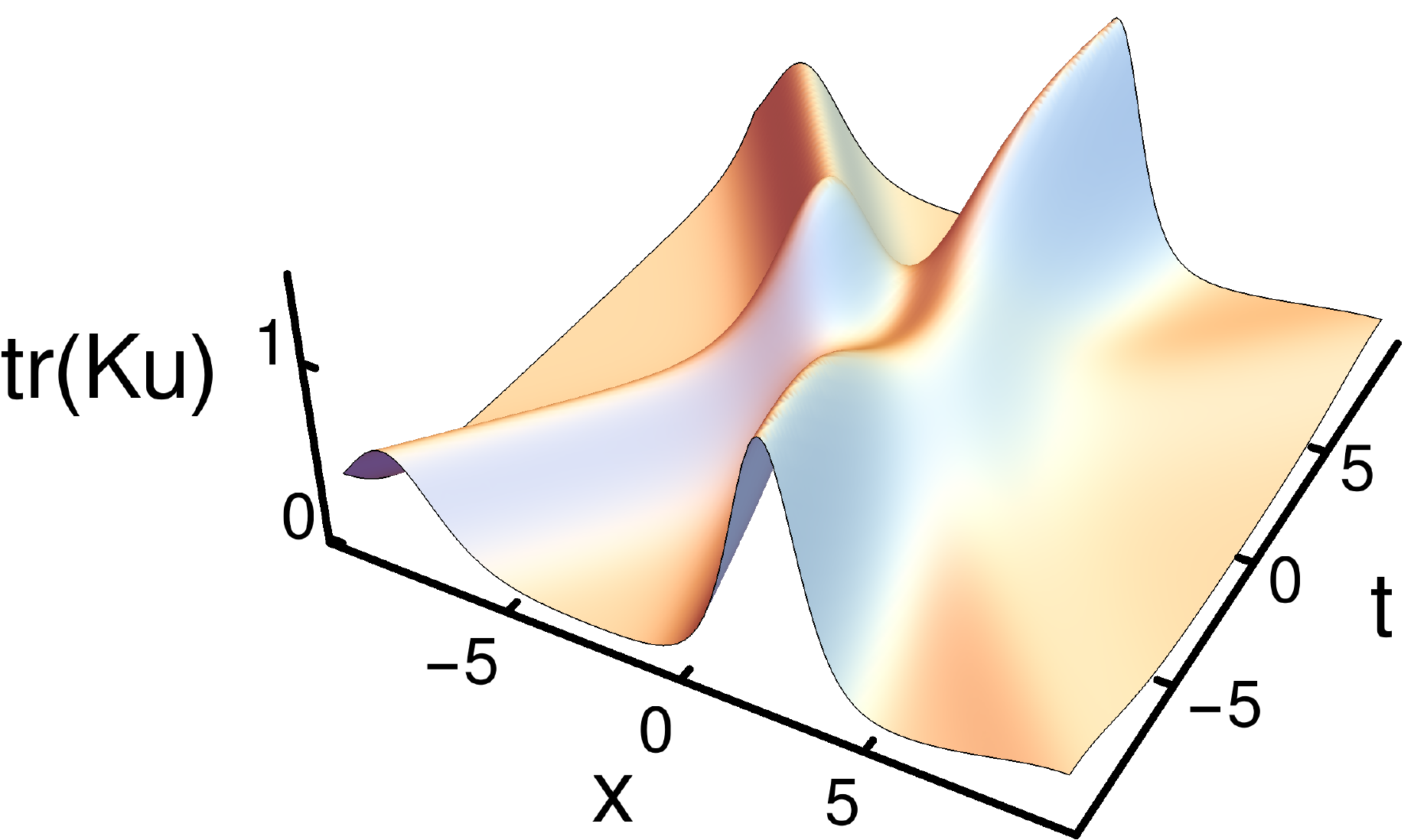} 
\hspace{1cm}
\includegraphics[scale=.18]{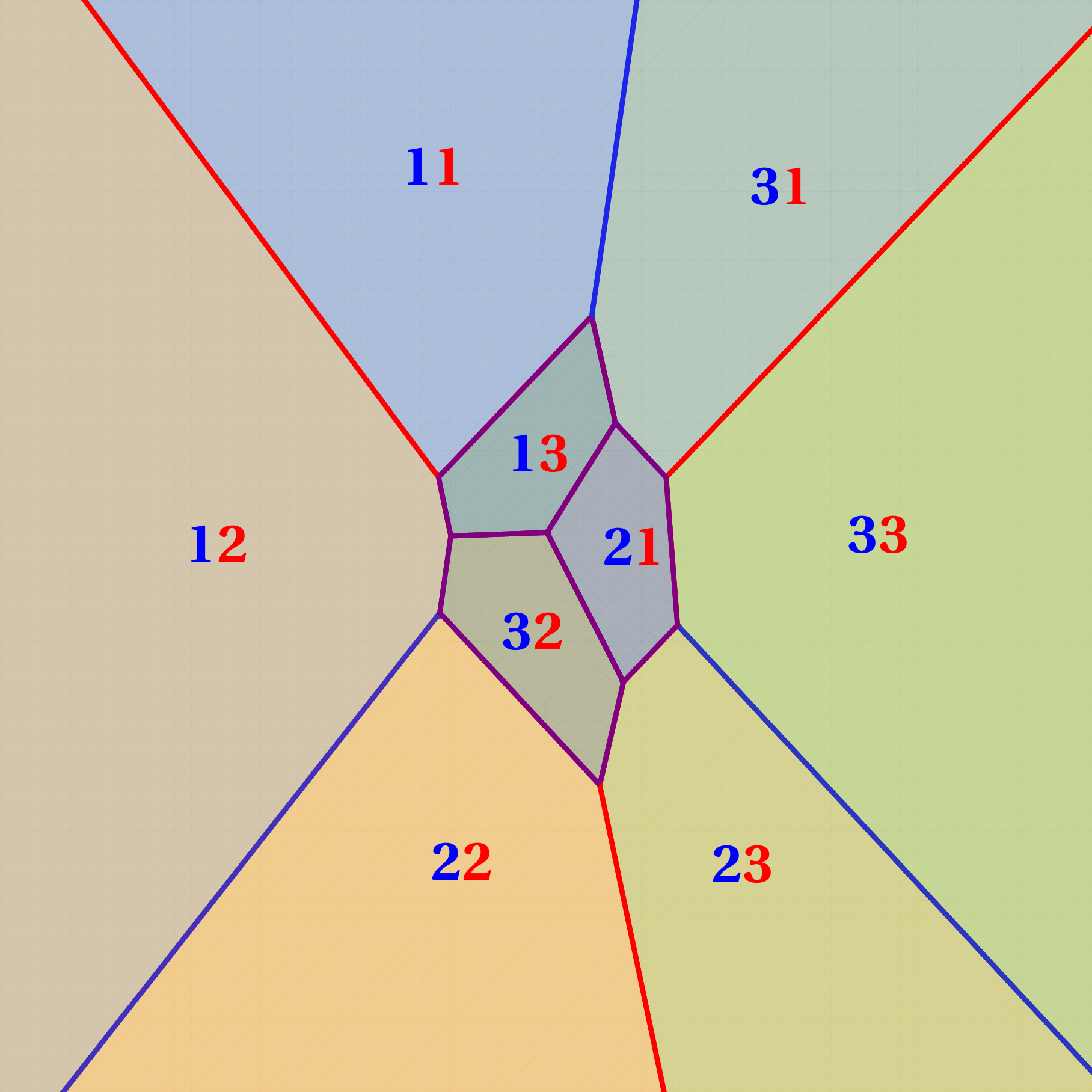}
\end{center}
\caption{Plot of $\mathrm{tr}(K u) = 2 (\log \tau)_{xx}$ and tropical limit graph the $3 \times 3$ matrix 
solution as specified in Example~\ref{ex:C'=!-C_3x3}. This shows a 3-soliton solution solution where none of 
the outgoing solitons has a velocity equal to that of one of the incoming solitons. 
The tropical limit graph shows that this is actually a superposition of a splitting soliton (red Y-shaped) 
and two merging solitons (blue reversed Y-shaped graph). 
\label{fig:3x3} }
\end{figure}
\begin{figure} 
\begin{center}
\includegraphics[scale=.26]{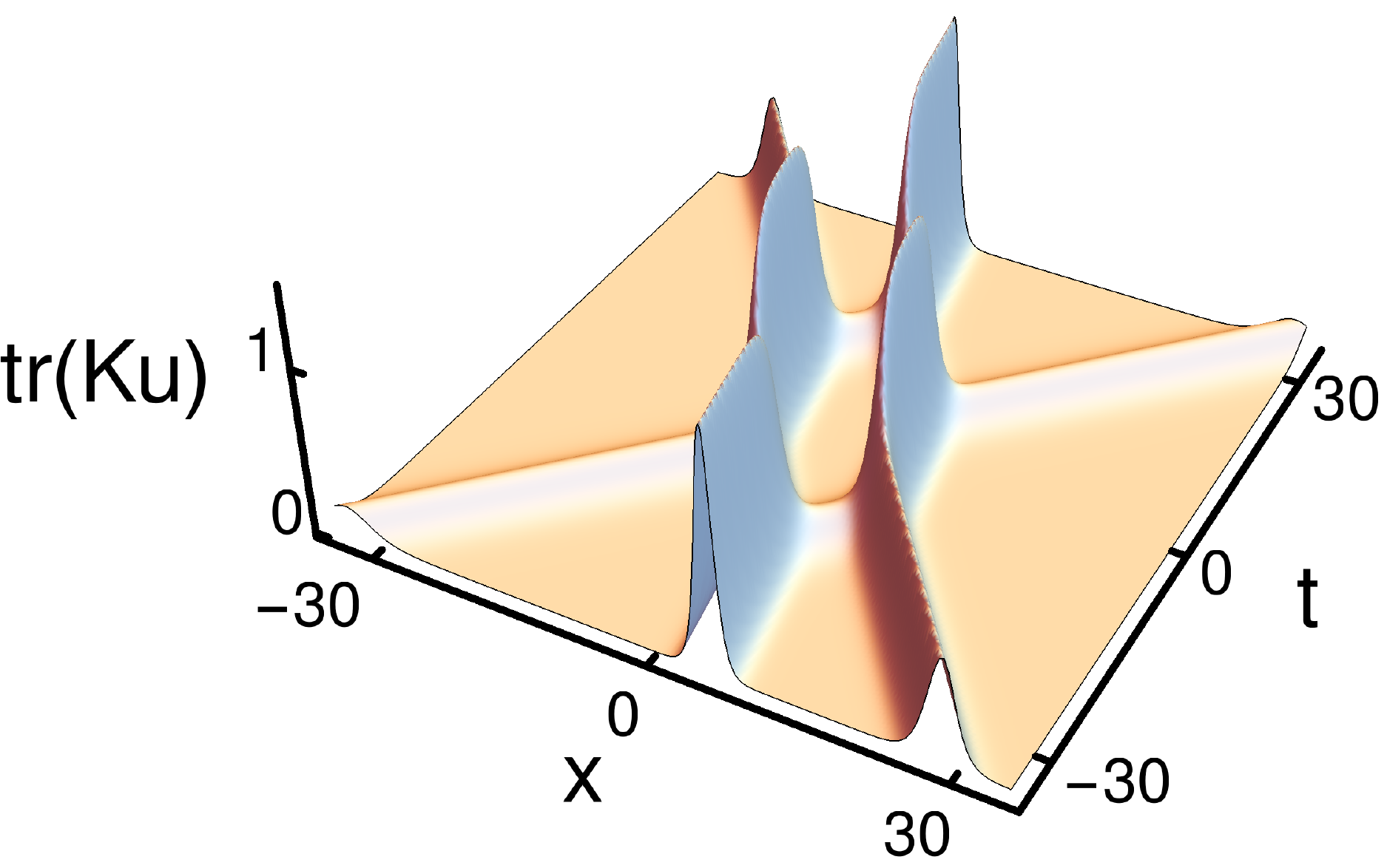} 
\hspace{.5cm}
\includegraphics[scale=.18]{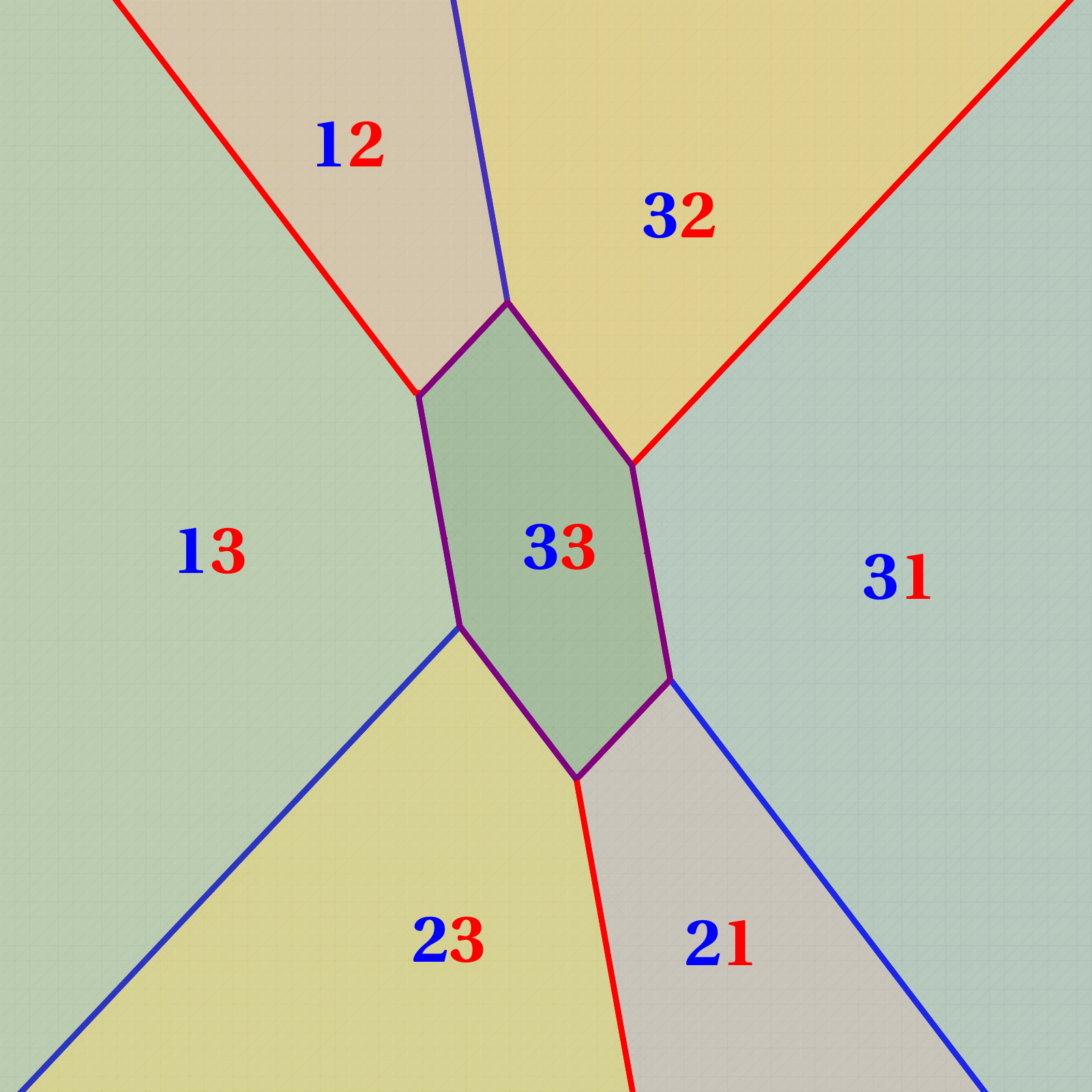} 
\end{center}
\caption{Plot of $\mathrm{tr}(K u) = 2 (\log \tau)_{xx}$ and tropical limit graph of a $3 \times 3$ matrix 
solution, as given in Section~\ref{sec:C'_noteq_-C}, with $N=1$ and 
$\theta_1 = (-1, 1, 1)^T$, $\theta_2 = (-2, 1, 1)$, $\theta_3 = (-1, -1,-1)$,  
$\lambda_1 = 1/10$, $\nu_1 = 1/10 + 10^{-5}$. Only seven of the nine phases, appearing in the 
corresponding function $\tau$, are visible.  
\label{fig:3x3_2} }
\end{figure}
\end{example}

\section{Conclusions}
\label{sec:concl}
The scalar Boussinesq equation exhibits richer soliton interactions than the KdV equation. 
This concerns head-on collisions and inelastic scattering.  
All this is nicely revealed in the tropical limit. Here the function $\tau$ determines 
the solution via $u = 2 \, (\log \tau)_{xx}$. 

In the vector case ($m=1$ or $n=1$), 
we have $\mathrm{tr}(K u) = 2 \, (\log \tau)_{xx}$, and this is a solution of the 
scalar Boussinesq equation. Specifying initial polarizations for sufficiently large negative 
time $t$, i.e., those of the incoming solitons, the distribution of polarizations over 
the tropical limit graph is obtained, for the class of solutions treated in Section~\ref{sec:C'=-C},  
by use of a Yang-Baxter $R$-matrix. For a more general class of solutions, also a tetragon map 
is at work, cf. Section~8 in \cite{DMH18p}. 

For the matrix ($m,n >1$) Boussinesq equation, $\mathrm{tr}(K u) = 2 \, (\log \tau)_{xx}$
is not in general a solution of the scalar equation. A \emph{non}linear Yang-Baxter map is 
at work, which is a reduction of the corresponding matrix KP Yang-Baxter map, 
recently obtained in \cite{DMH18}.  

We should note that a tropical limit of a matrix soliton solution can only be expected 
if we arrange the parameters such that exponentials absent in $\tau$ are also absent in 
the numerator of $u$. 
Otherwise the solution will exhibit exponential growth (in some space-time direction), in which 
case the solution should no longer be called a soliton solution. 

In the scalar and vector case, inelastic collisions are subject to severe restrictions. 
As expected, there is more freedom in the full matrix case. 

We concentrated on \emph{regular} solutions, for which all summands of $\tau$ are positive, so that 
the tropical limit, defined via Maslov dequantization, of $\tau$ makes sense. We did not have to 
compute this limit directly, however, since instead it is easier to determine the dominating phase regions 
and then the respective boundaries, which constitute the tropical limit graph. 
The latter point of view allows a generalization 
of the tropical limit, which also applies to \emph{singular} solutions, as explained at the end of 
Section~\ref{subsec:vBsq}. This allows us to determine, in a simple way, the locus of singularities 
of a solution. 

Nonlinear atomic or molecular chains, modeled by the scalar Boussinesq equation in a continuum limit, 
disregard a possible polarization of the particles. Taking polarizations or spins into account, 
such a chain might be described in the continuum limit by a vector or matrix version of 
the Boussinesq equation.  

Similar explorations, as done for matrix Boussinesq equations in this work, should be possible 
for other integrable equations too, in particular for discrete versions of the Boussinesq equation 
(see \cite{NPCQ92,Tong+Nijh05,Maru+Kaji10}, for example). 

\vspace{.3cm}
\noindent
\textbf{Acknowledgment.} X.-M.C. has been supported by the DAAD Research Grants - Short-Term Grants 2018 (57378443).

\renewcommand{\theequation} {\Alph{section}.\arabic{equation}}
\renewcommand{\thesection} {\Alph{section}}

\makeatletter
\newcommand\appendix@section[1]{
  \refstepcounter{section}
  \orig@section*{Appendix \@Alph\c@section: #1}
  \addcontentsline{toc}{section}{Appendix \@Alph\c@section: #1}
}
\let\orig@section\section
\g@addto@macro\appendix{\let\section\appendix@section}
\makeatother

\begin{appendix}

\section{Derivation of the binary Darboux transformation for the matrix Boussinesq equation}
\label{app:DBT}
We recall a binary Darboux transformation result of bidifferential calculus \cite{DMH13SIGMA,CDMH16}. 

\begin{theorem}
Let $(\Omega,\d, \bd)$ be a bidifferential calculus and $\Delta, \Gamma, \boldsymbol{\lambda}, \boldsymbol{\kappa}$ 
solutions of
\bez
  &&  \bd \Gamma = \Gamma \, \d \Gamma + [\boldsymbol{\kappa} , \Gamma] \, , \qquad
    \bd \boldsymbol{\kappa} = \Gamma \, \d \boldsymbol{\kappa} + \boldsymbol{\kappa}^2 \, ,  \\
  &&  \bd \Delta = (\d \Delta) \, \Delta - [\boldsymbol{\lambda} , \Delta] \, ,  \qquad
     \bd \boldsymbol{\lambda} = (\d \boldsymbol{\lambda}) \, \Delta - \boldsymbol{\lambda}^2 \, ,
\eez
and $\phi_0$ a solution of 
\be
      \d \bd \phi + \d \phi \, K \, \d \phi = 0 \, ,   \label{bDT_phi_eq}
\ee
where $\d K = 0 = \bd K$. 
Let $\theta$ and $\chi$ be solutions of the linear system
\be
     \bd \theta = (\d \phi_0) \, K \, \theta + (\d \theta) \, \Delta + \theta \, \boldsymbol{\lambda} \, ,
       \label{bDT_lin_sys}
\ee
respectively the adjoint linear system
\be
    \bd \chi = - \chi \, K \, \d \phi_0 + \Gamma \, \d \chi + \boldsymbol{\kappa} \, \chi \, .
       \label{bDT_adj_lin_sys}
\ee
Let $\Omega$ solve the compatible linear system
\be
  &&  \Gamma \, \Omega - \Omega \, \Delta = - \eta \, K \, \theta \, , \nonumber \\
  &&  \bd \Omega = (\d \Omega) \, \Delta - (\d \Gamma) \, \Omega + (\d \eta) \, K \, \theta
                   + \boldsymbol{\kappa} \, \Omega + \Omega \, \boldsymbol{\lambda} \, .   \label{bDT_Omega_eqs}
\ee
Where $\Omega$ is invertible, 
\be
     \phi = \phi_0 - \theta \, \Omega^{-1} \chi     \label{bDT_new_solution}
\ee
is a new solution of (\ref{bDT_phi_eq}). \hfill $\Box$
\end{theorem}

In the above theorem, we have to assume that all objects are such that the corresponding products are defined 
and that $\d$ and $\bd$ can be applied. 
Next we define a bidifferential calculus via
\bez
   && \d f = [\partial_x , f ] \, \zeta_1 + \frac{1}{2} [\partial_t + \partial_x^2 , f] \, \zeta_2  \, , \\
   && \bd f = \frac{1}{2} [\partial_t - \partial_x^2 , f] \, \zeta_1 
      + [\beta \, \partial_x - \frac{1}{3} \partial_x^3 , f] \, \zeta_2 \, , 
\eez
on the algebra $\cA = \cA_0[\partial_t,\partial_x]$, where $\cA_0$ is the algebra of smooth functions 
of two variables, $x$ and $t$, and $\partial_x$ is the operator of partial differentiation with 
respect to $x$. $\zeta_1, \zeta_2$ are a basis of a two-dimensional vector space $V$, from which we form 
the Grassmann algebra $\Lambda(V)$. $\d$ and $\bd$ extend to $\Omega = \cA \otimes \Lambda(V)$ 
in a canonical way, and to matrices with entries in $\Omega$. The equation (\ref{bDT_phi_eq}) 
is then equivalent to the matrix potential Bsq$_K$ equation (\ref{matrixBSq_K}). 
Choosing a solution $\phi_0$ and setting
\bez
     \Delta = \Gamma = - \partial_x \, , \quad
     \boldsymbol{\lambda} = - \frac{1}{3} C \, \zeta_2 \, , \quad
     \boldsymbol{\kappa} = - \frac{1}{3} C' \, \zeta_2 \, ,
\eez
the linear system (\ref{bDT_lin_sys}) and the adjoint linear system (\ref{bDT_adj_lin_sys}) 
lead to (\ref{Bsq_lin_sys}) and (\ref{Bsq_adj_lin_sys}), respectively. 
Furthermore, (\ref{bDT_Omega_eqs}) implies (\ref{Bsq_Omega}). According to the theorem, (\ref{bDT_new_solution}) 
yields a new solution of the matrix potential Bsq$_K$ equation (\ref{matrixBSq_K}). 

\end{appendix}


\begin{thebibliography}{10}

\bibitem{Bous1872}
Boussinesq J 1872 Th\'eorie des ondes et des remous qui se propagent le long
  d'un canal rectangulaire horizontal, et communiquant au liquide contenu dans
  ce canal des vitesses sensiblement pareilles de la surface au fond {\em J.
  Math. Pures Appl. 2$^e$ s\'erie\/} {\bf 17} 55--108

\bibitem{Pnev85}
Pnevmatikos S 1985 {\em Singularities \& Synamical Systems\/} ({\em
  North-Holland Mathematics Studies\/} vol 103) ed Pnevmatikos S (Elsevier
  Science Publishers B.V. (North Holland)) pp 397--437

\bibitem{Dick03}
Dickey L 2003 {\em Soliton Equations and Hamiltonian {S}ystems\/}
  (Singapore: World Scientific)

\bibitem{Bogd+Zakh02}
Bogdanov L and Zakharov V 2002 The Boussinesq equation revisited {\em Physica
  D\/} {\bf 165} 137--162

\bibitem{Rasi+Schi17}
Rasin A and Schiff J 2017 B\"acklund transformations for the Boussinesq
  equation and merging solitons {\em J. Phys. A: Math. Theor.\/} {\bf 50}
  325202

\bibitem{FST83}
Fal'kovich G, Spector M and Turitsyn S 1983 Destruction of stationary solutions
  and collapse in the nonlinear string equation {\em Phys. Lett. A\/} {\bf 99}
  271--274

\bibitem{Lamb+Muse89}
Lambert F and Musette M 1989 On soliton instabilities for the nonlinear string
  equation {\em Journal de Physique Colloques\/} {\bf 50 (C3)} C3--33--C3--38

\bibitem{DMH18}
Dimakis A and M\"uller-Hoissen F 2018 Matrix KP: tropical limit and Yang-Baxter maps, 
 {\em Lett. Math. Phys.} \url{https://doi.org/10.1007/s11005-018-1127-3}

\bibitem{DMH18p}
Dimakis A and M\"uller-Hoissen F 2017 Matrix Kadomtsev-Petviashvili equation: tropical limit,
  Yang-Baxter and pentagon maps, {\em Theor. Math. Phys.} {\bf 196} 1164--1173

\bibitem{Vese03}
Veselov A 2003 Yang-Baxter maps and integrable dynamics {\em Phys. Lett. A\/}
  {\bf 314} 214--221

\bibitem{Gonc+Vese04}
Goncharenko V and Veselov A 2004 {\em New Trends in Integrability and Partial
  Solvability\/} ({\em NATO Science Series II: Math. Phys. Chem.\/} vol 132) ed
  Shabat A {\em et~al.\/} (Dordrecht: Kluwer) pp 191--197

\bibitem{Tsuch04}
Tsuchida T 2004 $N$-soliton collision in the Manakov model {\em Progr.
  Theor. Phys.\/} {\bf 111} 151--182

\bibitem{APT04IP}
Ablowitz M, Prinari B and Trubatch A 2004 Soliton interactions in the vector
  NLS equation {\em Inv. Problems\/} {\bf 20} 1217--1237

\bibitem{DMH13SIGMA}
Dimakis A and M\"uller-Hoissen F 2013 Binary Darboux transformations in
  bidifferential calculus and integrable reductions of vacuum Einstein
  equations {\em SIGMA\/} {\bf 9} 009

\bibitem{CDMH16}
Chvartatskyi O, Dimakis A and M\"uller-Hoissen F 2016 Self-consistent sources
  for integrable equations via deformations of binary Darboux transformations
  {\em Lett. Math. Phys.\/} {\bf 106} 1139--1179

\bibitem{DMH11KPT}
Dimakis A and M\"uller-Hoissen F 2011 KP line solitons and Tamari lattices
  {\em J. Phys. A: Math. Theor.\/} {\bf 44} 025203

\bibitem{Hiro+Ito83}
Hirota R and Ito M 1983 Resonance of solitons in one dimension {\em J. Phys.
  Soc. Japan\/} {\bf 52} 744--748

\bibitem{Koda10}
Kodama Y 2010 KP solitons in shallow water {\em J. Phys. A: Math. Theor.\/}
  {\bf 43} 434004

\bibitem{NPCQ92}
Nijhoff F, Papageorgiou V, Capel H and Quispel G 1992 The lattice
  Gel'fand-Dikii hierarchy {\em Inv. Problems\/} {\bf 8} 597--621

\bibitem{Tong+Nijh05}
Tongas A and Nijhoff F 2005 The Boussinesq integrable system: compatible
  lattice and continuum structures {\em Glasgow Math. J.\/} {\bf 47A} 205--219

\bibitem{Maru+Kaji10}
Maruno K~I and Kajiwara K 2010 The discrete potential Boussinesq equation and
  its multisoliton solutions {\em Appl. Anal.\/} {\bf 89} 593--609

\end{thebibliography}
\end{document}